\documentclass[a4paper,fleqn]{cas-sc}
\usepackage[authoryear,longnamesfirst]{natbib}

\def\tsc#1{\csdef{#1}{\textsc{\lowercase{#1}}\xspace}}
\tsc{WGM} \tsc{QE} \tsc{EP} \tsc{PMS} \tsc{BEC} \tsc{DE}

\usepackage{amssymb}
\usepackage{lipsum}
\usepackage{graphicx}%
\usepackage{multirow}%
\usepackage{tabularx}
\usepackage{amsmath,amssymb,amsfonts}%
\usepackage{amsthm}%
\usepackage{mathrsfs}%
\usepackage[title]{appendix}%
\usepackage{xcolor}%
\usepackage{textcomp}%
\usepackage{manyfoot}%
\usepackage{booktabs}%
\usepackage{algorithm}%
\usepackage{algorithmicx}%
\usepackage{algpseudocode}%
\usepackage{listings}
\usepackage{appendix}

\theoremstyle{thmstyleone}
\newtheorem{theorem}{Theorem}

\theoremstyle{thmstyletwo}

\theoremstyle{thmstylethree}

\begin{document}
	\let\WriteBookmarks\relax
	\def\floatpagepagefraction{1}
	\def\textpagefraction{.001}
	\shorttitle{Mathematical modeling of tumor-macrophage interactions}
	\shortauthors{H.F. Zhang et~al.}
	
	\title [mode = title]{Mathematical model of tumor-macrophage interactions: Elucidating the tumor-driven macrophage phenotype reprogramming}  
	
	\author[1]{Haifeng Zhang}[orcid=0009-0002-1223-9927]
	\ead{haifengzhang202212@163.com}
	\affiliation[1]{organization={School of Mathematical Sciences},
		addressline={Jiangsu University}, 
		city={Zhenjiang},
		postcode={212013}, 
		country={China}}
	
	\author[2]{Yipu Qu}[orcid=0009-0008-9779-1975]
	\ead{quyipu@tiangong.edu.cn}
	\affiliation[2]{organization={School of Mathematical Sciences},
		addressline={Tiangong University}, 
		city={Tianjin},
		postcode={300387}, 
		country={China}}
	
	\author[3]{Wuyue Yang}[orcid=0000-0002-8270-8840]
	\ead{yangwuyue@bimsa.cn}
	\affiliation[3]{organization={Beijing Institute of Mathematical Sciences and Applications},
		city={Beijing},
		postcode={101408}, 
		country={China}}
	
	\author[2]{Chenghang Li}[orcid=0009-0008-4290-221X]
	\cormark[1]
	\ead{lichenghang@tiangong.edu.cn}
	\cortext[cor1]{Corresponding author}
	
	\begin{abstract}
		The interplay between tumor cells and macrophages plays a central regulatory role in cancer progression. In this study, we developed a mathematical model that incorporates tumor cells, M1-type macrophages, M2-type macrophages and an M3-type macrophage population characterized by dual phenotypic features. First, we analyzed the fundamental mathematical properties of the model and derived the conditions under which the system attains a tumor‑free stable state or a coexistence state of tumor and immune cells. Second, global sensitivity analysis revealed that key parameters governing macrophage polarization and intercellular communication vary dynamically during tumor development. Bifurcation analysis further identified the polarization rate of M1‑type macrophages ($\kappa$) and the baseline level of resting macrophages ($M_0$) as critical determinants of the system's dynamical behavior. Notably, using approximate Bayesian computation for parameter inference and dynamic simulations, the model successfully recapitulated the evolutionary trajectories of eight tumor samples. The results demonstrate that lower tumor burden is significantly associated with higher M1‑type macrophage infiltration and delayed peak time of M3‑type macrophage activation. Moreover, survival analysis indicated that both enhanced M1‑type macrophage infiltration and delayed peak time of M3‑type macrophage activation are correlated with longer survival time. In summary, this study not only provides a theoretical framework for understanding the dynamic mechanisms underlying tumor‑macrophage interactions but also proposes two potential clinical prognostic markers: the level of M1‑type macrophage infiltration and the peak time of M3‑type macrophage activation.
	\end{abstract}
	
	
	\begin{keywords}
		Macrophage polarization \sep Mathematical model \sep Virtual cohort \sep Approximate Bayesian computation \sep Bifurcation analysis \sep Sensitivity analysis
	\end{keywords}
	
	\maketitle
	
	\section{Introduction}
	
	Tumor-associated macrophages (TAMs) represent one of the most abundant immune cell populations within the tumor microenvironment (TME) \citep{DeNardo.NatRevImmunol.2019}. As a highly heterogeneous population, macrophages exhibit striking functional duality. They can exert both anti-tumor surveillance and pro-tumor promotional effects, a property that is tightly linked to their dynamic phenotypic plasticity \citep{Kloosterman.Cell.2023}. Based on their activation states and functional properties, TAMs are broadly categorized into three key subsets: (1) M1-type macrophages ($M_1$), activated by interferon gamma (IFN-$\gamma$) or lipopolysaccharide (LPS). They inhibit tumor growth by secreting pro-inflammatory cytokines such as IL-12 and tumor necrosis factor-alpha (TNF-$\alpha$) to stimulate immune responses, and through direct cytotoxic effects \citep{Lawrence.NatRevImmunol.2011,Noy.Immunity.2014,Xu.SigTransductTargetTher.2025}. (2) M2-type macrophages ($M_2$), induced by IL-4, IL-10, or IL-13. These cells promote tumor progression by driving angiogenesis, secreting growth factors like vascular endothelial growth factor (VEGF), transforming growth factor  $\beta$ (TGF-$\beta$) and suppressing adaptive immunity \citep{Lawrence.NatRevImmunol.2011,Noy.Immunity.2014,Xu.SigTransductTargetTher.2025}. (3) M3-type macrophages ($M_3$), an intermediate state with hybrid functions. They retain the capacity to differentiate toward either M1 or M2 lineages in response to microenvironmental signals \citep{Dunsmore.SciImmunol.2024}. The dynamic interconversion among macrophage phenotypes actively drives tumor progression by enabling immune evasion and adaptation to the microenvironment \citep{Biswas.NatImmunol.2010,Locati.AnnuRevPathol.2020}. Hence, elucidating how tumors and macrophages interact is critical for understanding the progression and evolutionary dynamics of cancer.
	
	In addition, treatment resistance is widely recognized as a major cause of therapeutic failure in oncology. As reported, TAMs can mediate resistance to anti-cancer therapies \citep{Xiao.FrontImmunol.2021,Huang.CritRevOncolHematol.2024}. Meanwhile, macrophage-targeted therapeutic strategies have the potential to overcome treatment resistance and improve clinical outcomes \citep{Mantovani.NatRevClinOncol.2017,Xu.SigTransductTargetTher.2025}. Therefore, a detailed understanding of the tumor-macrophage crosstalk serves as a critical bridge between deciphering cancer evolution and paving the way for innovative treatments. In this work, we employed a four-dimensional mathematical model to elucidate the influence of tumor-macrophage crosstalk on cancer progression and to identify potential targets for therapeutic intervention.
	
	With the rapid development of tumor immunology, a series of mathematical models have been developed to elucidate the dynamic regulatory mechanisms governing tumor-macrophage interactions. Breems et al. \citep{Breems.JTheorBiol.2016} incorporated the dynamics of Th1/Th2 cells and M1/M2 macrophages into a unified mathematical framework, thereby enabling an integrated analysis of the interplay between innate and adaptive immunity. Building upon this framework, Eftimie et al. \citep{Eftimie.JTheorBiol.2017} investigated the role of CD4+ T cell-macrophage interactions in the context of melanoma immunotherapy. Their findings indicated that tumors can be eliminated under either type I (Th1/M1-dominated) or type II (Th2/M2-dominated) immune responses, whereas tumor progression is generally associated with a type II response. In a subsequent study, Eftimie et al. \citep{Eftimie.ActaBiotheor.2019} employed a mathematical model to investigate tumor-immune dynamics during oncolytic virotherapy, identifying M2-to-M1 macrophage repolarization as a critical determinant for therapeutic efficacy. Recently, Shu et al. \citep{Shu.ApplMathModel.2020} introduced a simplified three-dimensional mathematical model, including tumor cells, M1-type and M2-type macrophages. Through Hopf bifurcation analysis, they demonstrated that the system can exhibit oscillatory behaviors, reflecting biological scenarios such as long-term tumor-immune coexistence and recurrence. Additionally, Guo et al. \citep{Guo.DCDSB.2023} further investigated the therapeutic implications of M2-type macrophage repolarization in cancer treatment. Collectively, these studies utilize ordinary differential equation models to simulate tumor-macrophage interactions and provide a theoretical basis for the development of innovative anticancer strategies.
	
	
	In recent years, multiscale mathematical models have emerged as powerful tools to quantitatively characterize the dynamic interactions between tumors and macrophage phenotypic heterogeneity. Lampropoulos et al. \citep{Lampropoulos.BullMathBiol.2025} developed a multiphase fluid dynamics model that incorporates macrophage polarization, molecular networks and TGF-$\beta$ targeting. The study revealed distinct spatial distribution: M1-type macrophages predominantly localized at the tumor periphery, whereas M2-type macrophages accumulated in the hypoxic core. Furthermore, blockade of TGF-$\beta$ receptors was shown to inhibit the M1-to-M2 phenotypic transition. Concurrently, Zheng et al. \citep{Zheng.MolCancerTher.2018} constructed a spatiotemporal model revealing macrophage-mediated mechanisms of treatment resistance in glioma immunotherapy. Building on this, Lin et al. \citep{Lin.MultiscaleModelSim.2025} developed a multiscale model integrating molecular, cellular and microenvironmental dynamics to elucidate mechanisms of immunotherapy resistance in glioblastoma mediated by tumor-macrophage crosstalk. Meanwhile, Liu et al. \citep{Liu.SciAdv.2025} introduced a multiscale mathematical model-informed reinforcement learning (M4RL) framework to simulate tumor-macrophage dynamics and optimize combination drug regimens. Together, these studies represent significant advances in the field of tumor immunotherapy, providing critical  theoretical and computational frameworks for deciphering macrophage-driven resistance and informing improved therapeutic strategies.
	
	
	Recently, single-cell transcriptomic lineage analysis has identified an intermediate macrophage population with the capacity for bidirectional transition between classical M1 and M2 states \citep{Dunsmore.SciImmunol.2024}. This discovery fundamentally challenges the traditional M1/M2 dichotomy and uncovers a previously underrecognized dimension of macrophage phenotypic plasticity. Such plasticity is therefore critical to explaining the functional diversity and contextual adaptability of macrophages in physiological homeostasis, pathological progression and therapeutic contexts. To investigate the dynamics of continuous macrophage phenotype switching, Eftimie et al. \citep{Eftimie.MathModMethAppliS.2020,Eftimie.MathBiosci.2020} proposed a mathematical model grounded in active particle dynamics theory. Subsequently, Eftimie et al. \citep{Eftimie.JTheorBiol.2021,Bartha.JTheorBiol.2022} further investigated the interactions among M1, M2, and intermediate macrophage phenotypes. To the best of our knowledge, the majority of existing mechanistic models in this field incorporate only the M1 and M2 macrophage subtypes. Beyond the studies referenced above, few studies have investigated the collective impact of three distinct macrophage populations on tumor dynamics within a unified modeling framework. However, macrophages exhibit functionally divergent roles in tumor regulation and the cooperative interplay of three macrophage subtypes shapes tumor evolution. Thus, M3-type macrophages should be explicitly integrated as a discrete population in such models. Elucidating this potential regulatory network will provide critical insights into the mechanisms underlying cancer elimination and progression.
	
	
	In this study, we developed a mathematical model that incorporates tumor cells along with M1-type, M2-type and a hybrid M3-type macrophage population to investigate the dynamics of tumor-macrophage interactions. The key findings are as follows: (1) Theoretical analysis established a critical threshold for the tumor growth rate ($\alpha^*$). When the actual growth rate falls below $\alpha^*$, the system converges to a stable tumor-free equilibrium, representing successful immune-mediated clearance; otherwise, a tumor-persistent equilibrium emerges. (2) Global sensitivity analysis uncovered a time-dependent regulatory logic: early tumor control is primarily driven by the proliferation rate ($\alpha$), whereas later dynamics are dominated by the tumor carrying capacity ($K$). (3) Bifurcation analysis revealed bistability between tumor-free and high-tumor states, offering a mechanistic explanation for the heterogeneous progression patterns observed in cancer. (4) Using approximate Bayesian computation, we clustered the \textit{in silico} mouse cohort into distinct phenotypic groups: high, medium, and low tumor burden. Lower tumor burden was strongly predicted by greater cumulative M1-type macrophage activity and a delayed peak in the M3-type macrophage population. (5) Survival analysis confirmed that both elevated M1-type macrophage activity and a delayed M3-type macrophage peak correlate with prolonged survival, highlighting their potential as prognostic biomarkers. Collectively, the model provides a quantitative framework for linking macrophage phenotypic plasticity to tumor dynamics.
	
	
	The paper is organized as follows. In Section \ref{sec:sec2}, we established a mathematical model to decipher the underlying mechanisms governing the dynamic relationship between tumor cells and macrophages. In Section \ref{sec:sec3}, we performed theoretical analysis and numerical analysis to investigate the dynamics of the proposed model. The study included: (1) analysis of the existence and stability of steady states; (2) model validation and parameter sensitivity analysis; (3) bifurcation analysis of the polarization rate of M1‑type macrophages and the baseline level of resting macrophages; (4) population dynamics analysis of tumor-macrophage interaction via approximate Bayesian computation; (5) the influence of macrophage heterogeneity on the proposed model; (6) survival analysis via particle swarm optimization. Finally, in Section \ref{sec:sec4}, we briefly discussed our conclusions.
	
	\section{Mathematical model}\label{sec:sec2}
	
	Traditionally, research on tumor-associated macrophages (TAMs) has largely on two polarized states: the pro-inflammatory M1-type phenotype and the pro-tumorigenic M2-type phenotype. However, emerging evidence indicates that TAMs can adopt a hybrid phenotype, referred to herein as M3-type macrophages, which co-expresses functional features of both M1-type and M2-type macrophages \citep{Sunil.SciTranslMed.2019,Watanabe.JClinInvest.2019,Dunsmore.SciImmunol.2024}. This phenotypic heterogeneity complicates a unified understanding of macrophage plasticity and poses challenges for predicting the efficacy of therapeutic strategies targeting TAMs. To quantitatively investigate the role of macrophages in tumor development, we developed a mathematical model of tumor-macrophage interactions, as illustrated in Fig. \ref{Fig1}. 
	
	\begin{figure}[htbp]
		\centering
		\centerline{\includegraphics[width=0.55\textwidth]{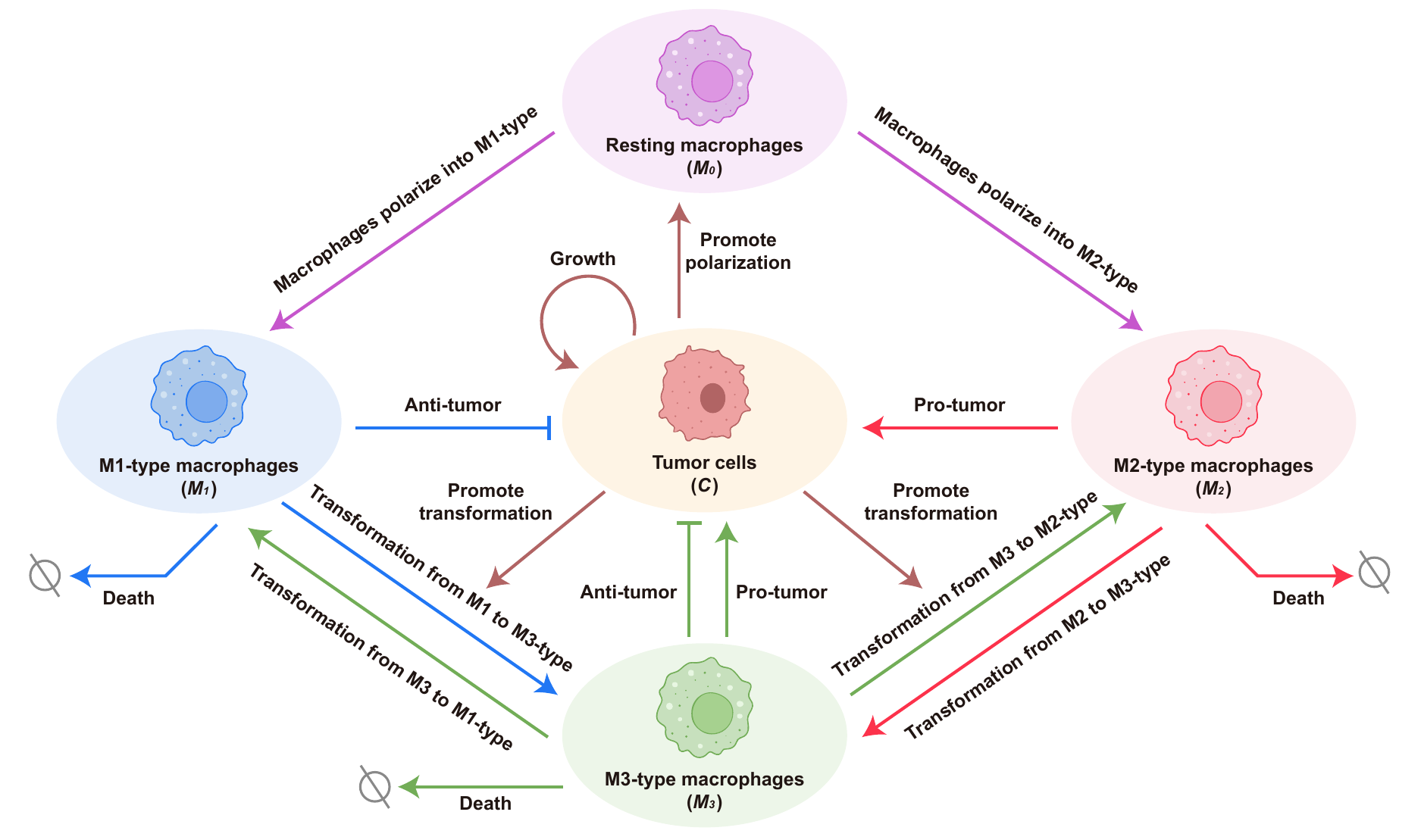}}
		\caption{\textbf{The dynamic regulatory network of the interaction between tumors and macrophages.} }
		\label{Fig1}
	\end{figure}
	
	The model is formulated as the following system of ordinary differential equations:
	\begin{equation}
		\label{Model}
		\left\{
		\begin{aligned}
			\frac{\mathrm{d} C}{\mathrm{d} t} & = \underbrace{\alpha \left \{ 1+\delta\left ( M_2+p M_3 \right )  \right \} \left ( 1 - \frac{C}{K} \right ) C}_{\text{Tumor growth promoted by M2/M3}} - \underbrace{\eta \left ( M_1 + (1-p)M_3 \right ) C,}_{\text{Killing tumor by M1/M3}}\\
			\frac{\mathrm{d} M_1}{\mathrm{d} t} & = \underbrace{\kappa M_0}_{\text{Polarization}} - \underbrace{\beta_{13}\frac{C}{K_1+C}M_1 + \beta_{31}M_3}_{\text{M1-to-M3 transformation}} - \underbrace{d_1 M_1}_{\text{Death}},\\
			\frac{\mathrm{d} M_2}{\mathrm{d} t} & = \underbrace{\gamma \left ( 1+\lambda \frac{C}{K_0+C} \right ) M_0}_{\text{Polarization}} + \underbrace{\beta_{32}\frac{C}{K_2+C}M_3- \beta_{23}M_2}_{\text{M3-to-M2 transformation}}  - \underbrace{d_2 M_2}_{\text{Death}}, \\
			\frac{\mathrm{d} M_3}{\mathrm{d} t} & =  \underbrace{\beta_{13}\frac{C}{K_1+C}M_1 - \beta_{31}M_3}_{\text{M1-to-M3 transformation}} - \underbrace{\beta_{32}\frac{C}{K_2+C}M_3 + \beta_{23}M_2}_{\text{M3-to-M2 transformation}} - \underbrace{d_3 M_3}_{\text{Death}}.
		\end{aligned}
		\right.
	\end{equation}
	Here, $C(t)$ represents the number of tumor cells at time $t$. The variables $M_1(t)$, $M_2(t)$ and $M_3(t)$ denote the number of pro-inflammatory M1-type, anti-inflammatory M2-type and intermediate M3-type macrophages at time $t$, respectively. This modeling framework facilitates a systematic analysis of how phenotypic transitions among macrophage subtypes collectively shape tumor growth dynamics and influence potential treatment outcomes.
	
	The model is formulated based on the following biological assumptions and parameter descriptions:
	\begin{itemize}
		\setlength{\itemsep}{2pt}
		\setlength{\parskip}{2pt} 
		\item Typically, tumor growth is increasingly limited by resource availability, such as oxygen and nutrients, as the tumor cell population expands, leading to a progressive decline in the proliferation rate as the cell population increases due to self-competition \citep{West.PNAS.2019}. To account for this saturation effect, we adopt the logistic growth law to characterize tumor growth dynamics \citep{West.PNAS.2019}. Here, the parameters $\alpha$ and $K$ represent the tumor's growth rate and carrying capacity, respectively.
		\item The anti-tumor activity of M1-type macrophages is mediated through direct mechanisms, including phagocytosis of tumor cells and the release of cytotoxic mediators \citep{Lawrence.NatRevImmunol.2011,Noy.Immunity.2014}. The parameter $\eta$ indicates the killing rate of tumor by M1-type macrophages. In contrast, M2-type macrophages promote tumor progression through multiple mechanisms, primarily by facilitating angiogenesis, secreting growth factors, and shaping an immunosuppressive microenvironment \citep{Lawrence.NatRevImmunol.2011,Noy.Immunity.2014}. The parameter $\delta$ indicates the tumor amplification coefficient mediated by M2-type macrophages. Following recent biological observation \citep{Dunsmore.SciImmunol.2024}, the M3-type macrophage population is subdivided based on functionality: a proportion $p$ exerts tumor-promoting effects, while the remaining proportion $1-p$ contributes to tumor suppression.
		\item According to the classical immunological paradigm, macrophages are viewed as terminally differentiated cells maintained by bone marrow-derived monocytes \citep{Lawrence.NatRevImmunol.2011,Guan.SignalTransductTargetTher.2025}. These monocytes enter tissues via circulation and differentiate into resting macrophages ($M_0$), which subsequently polarize into distinct functional subtypes in response to local microenvironmental signals \citep{Lawrence.NatRevImmunol.2011,Guan.SignalTransductTargetTher.2025}. In the model, $\kappa$ and $\gamma$ quantify the polarization rates toward M1-type and M2-type macrophages, respectively. Tumor cells are assumed to enhance the polarization of resting macrophages toward the M2-type. $\lambda$ denotes the polarization amplification coefficient of M2-type macrophages by tumor cells. The parameter $K_0$ represents the tumor half-saturation constant during the M2-type macrophages polarization process.
		\item Macrophages demonstrate considerable phenotypic plasticity \citep{Li.JImmunotherCancer.2021}. These phenotypic transitions are modeled by a set of rate parameters $\beta_{ij} ~ (i, j = 1 ~{\rm or}~ 2 ~{\rm or}~ 3)$, where $\beta_{ij}$ models the transformation rate from M$i$-type to M$j$-type macrophages. The tumor promotes macrophage polarization towards the M2-type. The parameter $K_1$ denotes the tumor half-saturation constant during the transition from M1-type to M3-type macrophages, respectively. Similarly, the transition from M3-type to M2-type macrophages depends on tumor cell density through the half-saturation constant $K_2$.
		\item In line with common mathematical modeling practice, cell death is assumed to follow first-order kinetics. For a cell population $M_i ~ (i=1, 2, 3)$, this is represented by a linear term $d_i M_i$, where $d_i$ denotes the corresponding death rate.
	\end{itemize}

	\section{Results}\label{sec:sec3}
	
	\subsection{Theoretical analysis} \label{sec:sec3_1}
	
	In this section, unless otherwise stated, all parameters of system \eqref{Model} are taken to be positive. Here, we examined the fundamental mathematical properties of system \eqref{Model}. These properties are necessary to ensure biological plausibility and to analyze long-term behavior. First, we obtained that all solutions of \eqref{Model} corresponding to positive initial conditions remain positive for all time.
	
	\begin{theorem}
		\label{thm2.1} \rm{\textbf{(Positivity of solutions)}}
		The solutions of system \eqref{Model} are positive for $t \geq 0$ when the initial conditions satisfy $C(0) > 0$, $M_1(0) > 0$, $M_2(0) > 0$ and $M_3(0) > 0$.
	\end{theorem}
	
Next, we establish that all solutions of system \eqref{Model} originating from positive initial conditions remain bounded.
	
	\begin{theorem}
		\label{thm2.1_b} \rm{\textbf{(Boundedness of solutions)}}
		The solutions of system \eqref{Model} that start positive remain bounded.
	\end{theorem}
	
	The system \eqref{Model} exhibits two classes of biologically meaningful equilibria: tumor-free equilibria of the form $(C, M_1, M_2, M_3) = \left(0, M_{10}^*, M_{20}^*, M_{30}^* \right)$ and tumorous euqilibria of the form $(C, M_1, M_2, M_3) = (0, M_1^*, M_{2}^*, M_{3}^*)$, where $M_{i0}^* ~(i=1,2,3)$ are positive. We now establish the existence and uniqueness conditions for the tumor-free equilibrium.
	
	\begin{theorem}
		\label{thm2.1_c}\rm{\textbf{(Existence of tumor-free equilibrium)}}
		If $\beta_{23} \geqslant 0$, system \eqref{Model} has a unique tumor-free equilibrium $\left(0, M_{10}^*, M_{20}^*, M_{30}^* \right)$, where 
		\begin{equation}
			\label{tre0}
			\begin{aligned}
				M_{10}^* &= \dfrac{ \kappa M_0 (\beta_{23}+d_2) (\beta_{31}+d_3)  + \beta_{31} \gamma M_0  \beta_{23} }{ d_1 (\beta_{23}+d_2) (\beta_{31}+d_3) }, \\
				M_{20}^* &= \dfrac{ \gamma M_0}{\beta_{23}+d_2}, \\
				M_{30}^* &= \dfrac{\gamma M_0  \beta_{23} }{(\beta_{23}+d_2) (\beta_{31}+d_3)}.
			\end{aligned}
		\end{equation}
	\end{theorem}
	
	We now examine the asymptotic stability of the tumor-free equilibrium and establish conditions for the existence of a positive equilibrium of system \eqref{Model}. Our analysis identifies a critical threshold condition that separates two distinct dynamical regimes: asymptotic stability of the tumor-free equilibrium, corresponding to tumor eradication, or existence of a positive equilibrium, indicative of tumor persistence.
	
	\begin{theorem}
		\label{thm2.1_d}
		\rm{\textbf{(Local stability of tumor-free equilibrium )}}
		If 
		\begin{equation}
			\label{eq:ps_4aaa}
			\alpha < \eta \left(  M_{10}^* +(1-p) M_{30}^* \right)/ \left(1+ \delta ( M_{20}^* + p  M_{30}^*) \right),
		\end{equation}
		the tumor-free equilibrium $\left( 0, M_{10}^*, M_{20}^*, M_{30}^* \right)$ is locally asymptotically stable,  where $M_{10}^*$, $M_{20}^*$ and $M_{30}^*$ are given by \eqref{tre0}.
	\end{theorem}
	
	We now turn to the existence of positive equilibria.
	\begin{theorem}
		\label{thm2.1_e}
		\rm{\textbf{(Existence of positive equilibria)}} If 
		\begin{equation}
			\label{eq:ps_4aa0}
			\alpha > \eta \left(  M_{10}^* +(1-p) M_{30}^* \right)/ \left(1+ \delta ( M_{20}^* + p  M_{30}^*) \right),
		\end{equation}
		the system \eqref{Model} has at least one positive steady state $S^* = (C^*, M_1^*, M_2^*, M_3^*)$, where $M_{10}^*$, $M_{20}^*$ and $M_{30}^*$ are given by \eqref{tre0}. 
	\end{theorem}
	
	Together, Theorems \ref{thm2.1}-\ref{thm2.1_e} provide a mathematical framework for characterizing tumor-macrophage interactions. The detailed proofs of Theorems \ref{thm2.1}-\ref{thm2.1_e} are provided in \textbf{Appendix A}. We first prove the positivity and boundedness of the model solutions (Theorems \ref{thm2.1} and \ref{thm2.1_b}), ensuring that all model trajectories remain within biologically plausible ranges. Furthermore, Theorems \ref{thm2.1_d} and \ref{thm2.1_e} are of particular biological relevance: they define an explicit critical threshold for the tumor growth rate:
	\begin{equation}
		\alpha^* = \eta \left(  M_{10}^* +(1-p) M_{30}^* \right)/ \left(1+ \delta ( M_{20}^* + p  M_{30}^*) \right). 
	\end{equation}
	This threshold quantifies the maximum tumor growth rate at which the immune system, constituted by tumor-associated macrophages, can still achieve complete tumor eradication. If the tumor growth rate $\alpha$ falls below $\alpha^*$, the immune response is sufficiently strong to eliminate all tumor cells, resulting in a locally asymptotically stable tumor-free equilibrium. In contrast, if $\alpha$ exceeds $\alpha^*$, tumor growth persists despite immune control. To gain further insight into the regulatory mechanisms governing tumor-macrophage crosstalk, we now turn to numerical simulations.
	
	\subsection{Model validation and parameter sensitivity analysis} \label{sec:sec3_3}
	
	To validate the proposed dynamical model, we inferred key parameters using a Markov Chain Monte Carlo (MCMC) framework. The inference was calibrated against experimental time‑course measurements of tumor volume obtained from study \citep{Reda.NatCommun.2022} (data usage and parameter estimation are provided in \textbf{Appendix B}). Model performance was quantified by the coefficient of determination $R^2$. With the baseline parameter set, our simulations achieved $R^2 = 0.988$. As shown in Fig. \ref{Fig2}A, the simulated tumor growth dynamics closely align with the experimental data, confirming the model’s ability to recapitulate the observed progression.
	
	\begin{figure}[htbp]
		\centering
		\centerline{\includegraphics[width=0.9\textwidth]{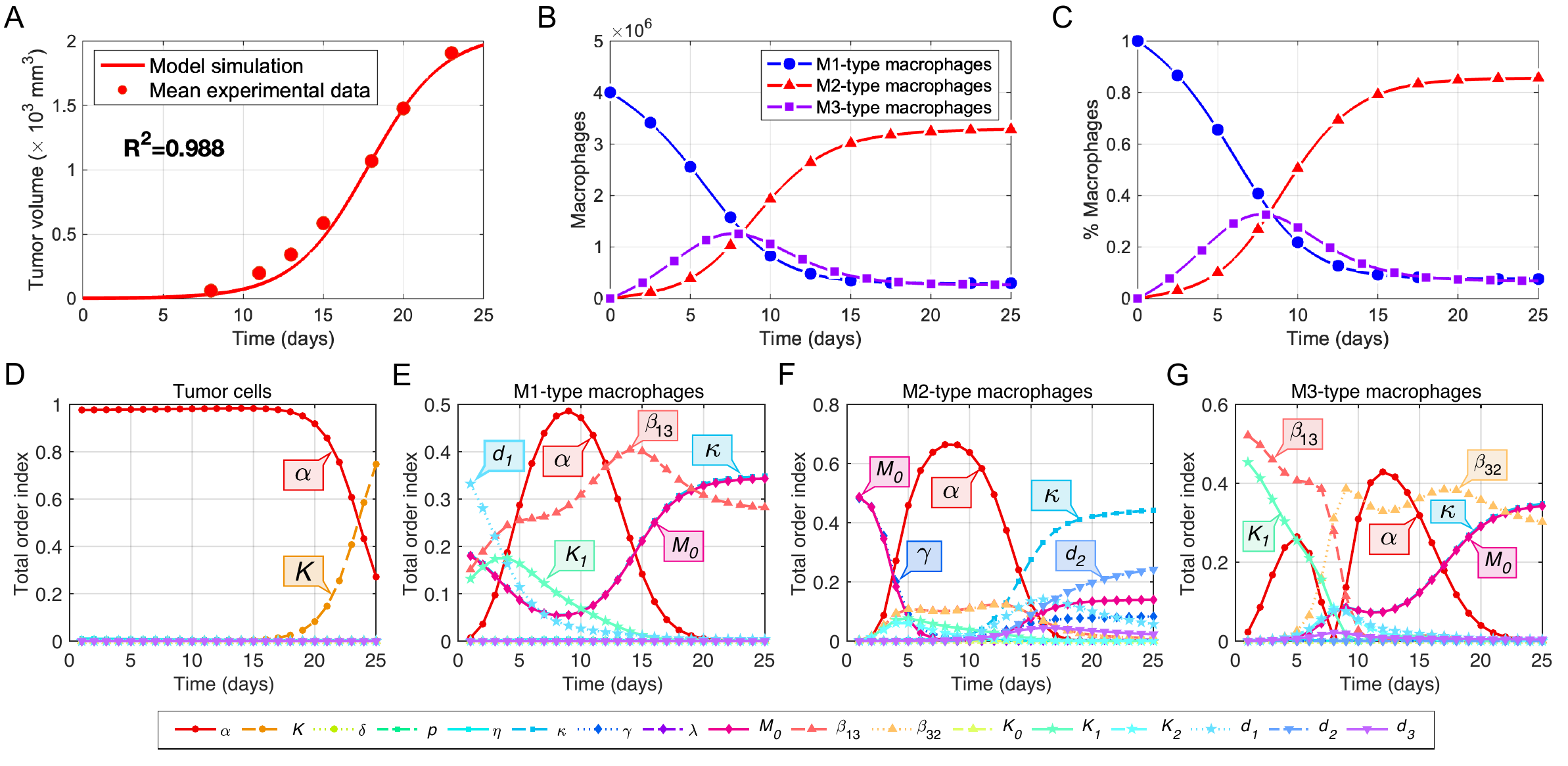}}
		\caption{\textbf{Model validation and global sensitivity analysis.} (A) Simulated tumor growth trajectories (solid line) closely match experimental data (dots), with $R^2 = 0.988$. (B) Temporal trajectories of M1‑type, M2‑type, and M3‑type macrophage subpopulations predicted by the calibrated model. (C) Temporal trajectories of the relative fractions of the three macrophage phenotypes. (D) - (G) Sobol' global sensitivity analysis quantifies the influence of individual parameters on tumor cell and macrophage phenotype dynamics.}
		\label{Fig2}
	\end{figure}
	
	The model further predicts the temporal dynamics of three major tumor-associated macrophages (TAMs) phenotypic states (Fig. \ref{Fig2}B). Simulations indicate an initial predominance of the anti-tumor M1-type macrophages, which gradually declines as the M3-type macrophage population increases, marking a transition toward M3-type macrophages. With continued tumor progression, M3-type macrophages further differentiate toward the pro-tumor M2 phenotype, leading to sustained accumulation of M2-type macrophages. The M3-type macrophages, partially sustained by the declining M1-type pool, peak transiently before diminishing, ultimately resulting in M2-type macrophages becoming the dominant subset in later stages. The early dominance of M1-type macrophages corresponds to initial anti-tumor immunity, whereas the subsequent rise of M3-type macrophages illustrates macrophage plasticity under evolving tumor signals. Finally, the late dominance of M2-type macrophages underscores their role in promoting immunosuppression, angiogenesis, and tissue remodeling in advanced disease. The simulated co-evolution of these subsets aligns with prior theoretical and experimental work \citep{Eftimie.JTheorBiol.2021, Dunsmore.SciImmunol.2024}, confirming that the model captures key features of tumor-macrophage crosstalk. Furthermore, temporal changes in
	the relative proportions of the three macrophage phenotypes mirror the trends in their absolute abundances (Fig. \ref{Fig2}C), reinforcing the internal consistency of the model dynamics.
	
	To identify the key drivers of macrophage polarization dynamics, we performed a parameter sensitivity analysis to evaluate how variations in input parameters influence model outputs. Specifically, we implemented the Sobol' method via Latin Hypercube Sampling (LHS) with a ±10\% parameter perturbation range to quantify the temporal evolution of parameter sensitivities. The total order index of Sobol' index is adopted as a measure of the overall influence of each input parameter. As illustrated in Fig. \ref{Fig2}D, the population of tumor cells is most sensitive to the tumor growth rate ($\alpha$) during early stages. As the disease progresses, the total order index of tumor carrying capacity ($K$) gradually increases and exceeds $\alpha$ on the approximately 24th day. This shift indicates that the dominant factor regulating cancer cell population dynamics transitions from the intrinsic tumor growth rate to the environmental carrying capacity as the tumor progresses.
	
	In the dynamics of M1-type macrophages, several parameters demonstrate notable influence over time, particularly the death rate ($d_1$), the transformation rate to M3-type macrophages ($\beta_{13}$), the tumor growth rate ($\alpha$), the polarization rate ($\kappa$), and the baseline level of resting macrophages ($M_0$) (Fig. \ref{Fig2}E). The sensitivity indices for $d_1$, $\kappa$, $M_0$ and $K_1$ progressively decline in the initial stage, while those for $\beta_{13}$ and $\alpha$ exhibit a non-monotonic trend of an initial increase followed by a decrease. Moreover, a resurgence in the sensitivity of M1-type macrophage polarization-related parameters ($\kappa$ and $M_0$) resumes from day 9 onward, reestablishing their role as principal determinants of M1-type macrophage dynamics in later phases. The temporal variations in parameter sensitivities suggest that the functional state of M1-type macrophages is dynamically regulated in a multi‑stage and multi‑factor manner. In the early phase, cell death ($d_1$) and initial polarization signals ($\kappa$, $M_0$) dominate population establishment. Subsequently, tumor growth ($\alpha$) and transition toward the M3‑type ($\beta_{13}$) emerge as competing regulatory factors within the microenvironment. The renewed rise in sensitivity of polarization-related parameters ($\kappa$, $M_0$) at later stages indicates that sustaining the anti‑tumor M1 phenotype during tumor progression may require persistent polarization signals to counteract phenotypic conversion pressures within the evolving microenvironment.
	
	Similarly, the dynamics of M2-type macrophages are shaped by key parameters including $M_0$, $\alpha$, $\kappa$, the death rate $d_2$ and the polarization rate $\gamma$ (Fig. \ref{Fig2}F). In the early phase, the influence of $\gamma$ and $M_0$ gradually declines, whereas the impact of the tumor growth rate $\alpha$ rises sharply and dominates M2 behavior. As the system progresses into later stages, $\kappa$ and $d_2$ emerge as the principal drivers of M2-type macrophage dynamics. This progression suggests that the dynamic regulation of M2-type macrophages occurs in distinct stages: it is primarily driven by tumor cells in the early phase, and gradually shifts to being governed by interactions among different macrophage subsets in later stages.
	
	Meanwhile, the sensitivity analysis for M3-type macrophages reveals a distinct temporal pattern (Fig. \ref{Fig2}G). The influence of M1-to-M3 transformation parameters ($\beta_{13}$ and $K_1$) gradually declines in the early phase. In contrast, the sensitivity to the M3-to-M2 transformation ($\beta_{32}$) and M1-type macrophages polarization parameters ($\kappa$ and $M_0$) increase over time, becoming the dominant factors governing M3-type macrophages dynamics in the later stage. Notably, the sensitivity of the M3 population to the tumor proliferation rate ($\alpha$) exhibits a bimodal pattern over time compared to M1 and M2-type macrophages. This indicates that the dynamics of M3-type macrophages are driven by a complex interplay of multidimensional transformation parameters and cross regulatory signals. The observed bimodal response to tumor growth may reflect a stage-dependent functional switch in M3-type macrophages during tumor progression, offering a new perspective for analyzing macrophage inter-subtype crosstalk and the heterogeneity of the tumor immune microenvironment.
	
	\subsection{Numerical analysis of the existence and stability of equilibria} \label{sec:sec3_2}
	
	In this section, to further elucidate the influence of tumor-associated macrophages on the long-term dynamics of the system described by Eq. \eqref{Model}, we conducted a numerical investigation into how the polarization rate of M1-type macrophages ($\kappa$) and baseline level of resting macrophages ($M_0$) affect the existence and stability of the system equilibria. Using the default parameter values provided in Table \ref{tab:table1}, we performed a quantitative bifurcation analysis and the results are presented in Fig. \ref{Fig_bif1} and \ref{Fig_bif2}. Fig. \ref{Fig_bif1}A demonstrates the dependence of the number of steady states on the parameters $\kappa$ and $M_0$. Distinct regions in the diagram correspond to different dynamical regimes: tumor-free equilibria (pink), bistability (sky-blue) and tumorous equilibria (wheat). Furthermore, Fig. \ref{Fig_bif1}B, C and \ref{Fig_bif2}A-H provide detailed bifurcation diagrams depicting the existence and stability of steady states with respect to specific parameter variations. In these plots, unstable and stable steady states are denoted by dashed and solid curves, respectively.
	
	\begin{figure}[htbp]
		\centering
		\centerline{\includegraphics[width=1.14\textwidth]{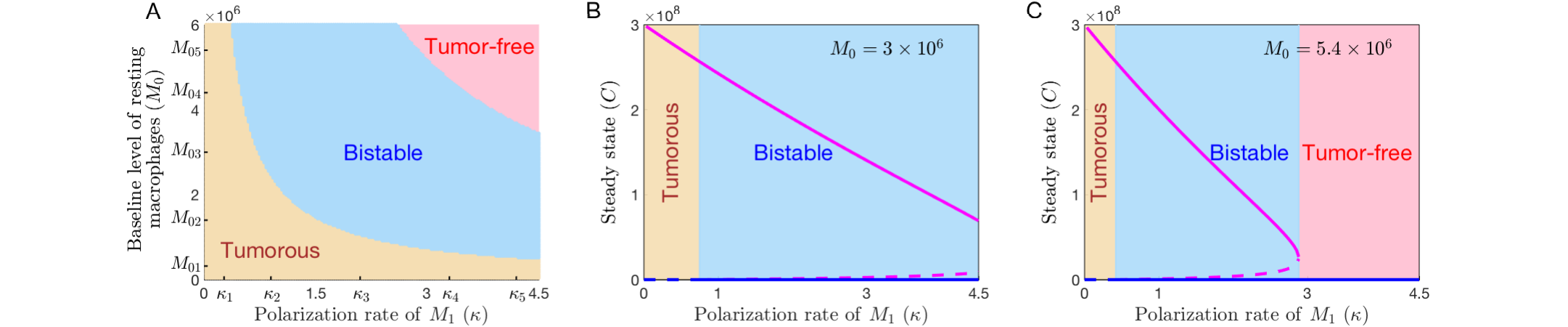}}
		\caption{\textbf{Bifurcation analysis of tumor-macrophage dynamics with respect to polarization rate of M1-type macrophages $\kappa$ and baseline level of resting macrophages $M_0$.} (A) Two‑parameter bifurcation diagram of tumor cell counts $(C)$  in the 
			($\kappa$, $M_0$) plane. Color‑coded regions correspond to distinct qualitative regimes: tumor-free equilibria (pink), bistability (sky-blue) and sustained tumor presence (wheat). (B)-(C) One‑parameter bifurcation diagrams of tumor cell counts as a function of $\kappa$ with different values of $M_0$ ($3 \times 10^6$ and $5.4 \times 10^6$ denoted $M_{03}$ and $M_{05}$ shown in (A), respectively). Dashed and solid curves represent unstable and stable steady states, respectively. Blue and magenta lines represent tumor-free equilibrium and positive steady states, respectively.}
		\label{Fig_bif1}
	\end{figure}
	
	As illustrated in Fig. \ref{Fig_bif1}A, the numerical simulations suggest that for sufficiently low levels of resting macrophages (e.g., $M_0 = M_{01}$), a unique positive steady state exists and remains asymptotically stable across all values of the polarization rate of M1-type macrophages ($\kappa$). This indicates that when the resting macrophage population is small, the polarization of these cells into the anti‑tumor M1-type macrophages is insufficient to counteract tumor growth. Consequently, immune escape prevails, and the system converges to a tumor‑dominant state irrespective of $\kappa$. As $M_0$ increases moderately (e.g., $M_0 = M_{02}$ or $M_{03}$), the number of positive equilibria depends on $\kappa$, yielding either a single tumorous steady state (wheat region) or two positive steady states (sky‑blue region). In the latter bistable scenario, the two stable attractors correspond to a tumor‑free equilibrium (representing a healthy state) and a high‑tumor equilibrium (representing disease progression), as shown in Fig. \ref{Fig_bif1}B. This bistability implies a critical transition: tumor control or progression becomes sensitive to initial conditions. For higher levels of resting macrophages $M_0$ ( $M_0 = M_{04}$ or $M_{05}$), the number of positive equilibria again varies with $\kappa$, ranging from zero to two. When two positive equilibria coexist, the state with higher tumor burden remains stable (Fig. \ref{Fig_bif1}C). This behavior may reflect the possibility that elevated resting macrophage counts promote polarization toward the pro‑tumor M2-type macrophages, thereby facilitating tumor persistence. However, when $\kappa$ is sufficiently large, only the tumor‑free equilibrium persists, indicating that a high polarization rate of M1-type macrophages is sufficiently large and can compensate for excessive resting macrophages and effectively suppress tumor growth. A similar dependence on $M_0$ is observed when $\kappa$ is held fixed at representative values $\kappa_i$ ($\kappa_i = 1, 2, ... ,5$), underscoring the coupled roles of resting macrophage abundance and polarization rate in determining the existence and stability of system equilibria. The corresponding results, which illustrate the shifts in equilibrium profiles as $M_0$ varies, are presented in detail in Fig. \ref{Fig_bif2}E-H.
	
	In addition, we conducted a series of numerical bifurcation analyses for discrete values of polarization rate $\kappa_i$ and the resting macrophage levels $M_{0i}$ (with $i=1,2,...,5$). The corresponding positive steady-state profiles are presented in Fig. \ref{Fig_bif2}. As demonstrated in Fig. \ref{Fig_bif2}A-D, considering a stable positive equilibrium and a fixed resting macrophage level ($M_0$), an increase in the M1-type macrophage polarization rate ($\kappa$) leads to a progressive decline in tumor cell numbers, while elevating the abundances of M1-type, M2-type and M3-type macrophages. Thus, for a fixed $M_0$, increasing $\kappa$ promotes polarization toward M1-type macrophages, expanding the M1 population. The expanded M1 pool subsequently transforms into M3-type macrophages, which in turn facilitates further conversion into the M2-type macrophages. Furthermore, Fig. \ref{Fig_bif2}E-H show that when the M1-type macrophage polarization rate ($\kappa$) is held, an increase in the resting macrophage level results in a gradual reduction in tumor cell count, together with elevated numbers of M1-type, M2-type and M3-type macrophages. This suggests that, at a fixed $\kappa$, rising resting macrophage level promotes polarization toward the M1 phenotype, boosting M1-type macrophage abundance. The resulting enlarged M1 population further transitions into M3-type macrophages, which subsequently enhances the conversion of M3-type into M2-type macrophages. Together, these findings elucidate a cascade of regulatory interactions among macrophage subtypes that orchestrate the dynamic crosstalk with tumor cells.
	
	\begin{figure}[htbp]
		\centering
		\centerline{\includegraphics[width=1.15\textwidth]{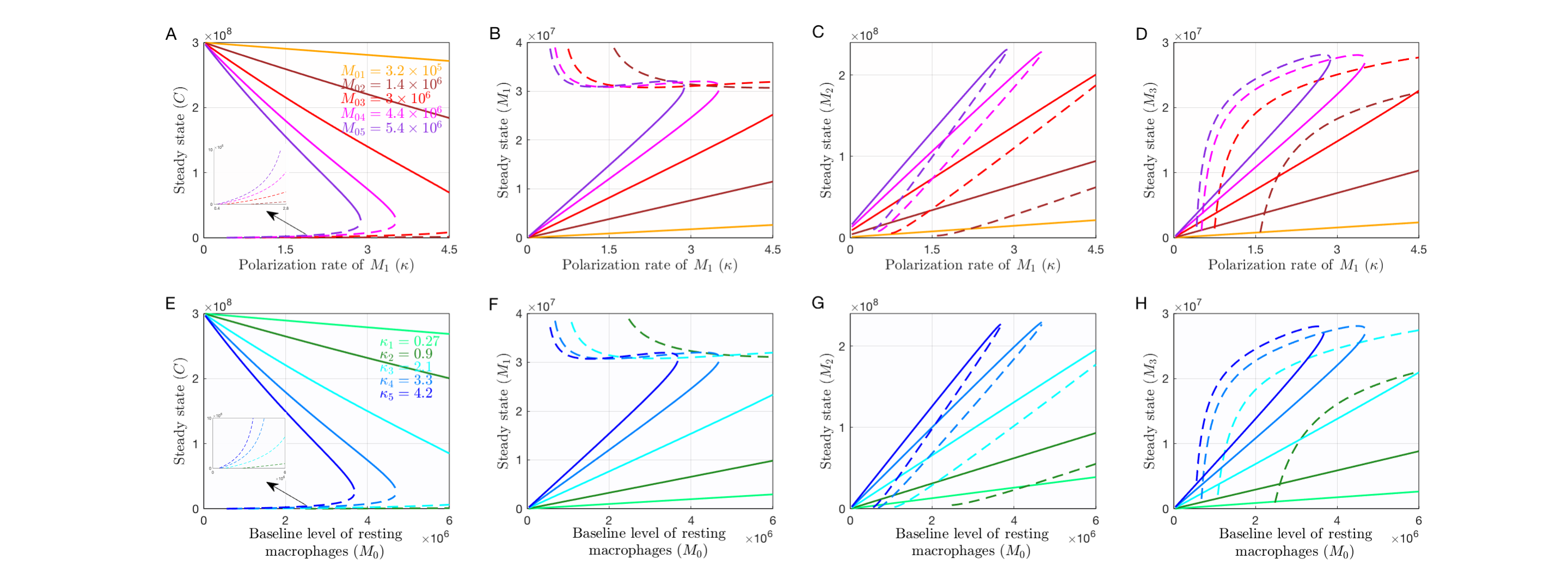}}
		\caption{\textbf{Bifurcation analysis with respect to polarization rate of M1-type macrophages $\kappa$ and baseline level of resting macrophages $M_0$.} (A)-(D) Bifurcation diagrams illustrating the dependence of the positive steady‑state counts of (A) tumor cells, (B) M1-type, (C) M2-type and (D) M3-type macrophages on $\kappa$ for five different values of $M_0$ ($M_{10}$, $M_{20}$, $M_{30}$, $M_{40}$ and $M_{50}$ in Fig. \ref{Fig_bif1}A), respectively. (E)-(H) Bifurcation diagrams showing the steady‑state counts of (E) tumor cells, (F) M1‑type, (G) M2‑type and (H) M3‑type macrophages as functions of $M_0$ for five fixed values of $\kappa$ ($\kappa_1$, $\kappa_2$, $\kappa_3$, $\kappa_4$ and $\kappa_5$ in Fig. \ref{Fig_bif1}A), respectively. Solid and dashed lines represent stable and unstable steady states, respectively.}
		\label{Fig_bif2}
	\end{figure}
	
	\subsection{Population dynamics analysis of tumor-macrophage interaction} \label{sec:sec3_4}
	
	To characterize individual-specific dynamical profiles, we employed the approximate Bayesian computation (ABC) method to generate \textit{in silico} parameter ensembles ($n$=100 per mouse) for a cohort of eight mice, informed by corresponding experimental time-series data. Model parameters were bounded within physiologically and mechanistically plausible intervals, as supported by established biological mechanisms and prior studies: $\alpha \in [0.38, 0.46]$, $K \in [2.17 \times 10^8, 5.00 \times 10^8]$, $M_0 \in [1 \times 10^5, 1 \times 10^6]$, $\kappa \in [0, 1]$, $\beta_{13} \in [0, 1]$, $\beta_{32} \in [0, 1]$, $\eta \in [1 \times 10^{-8}, 3 \times 10^{-8}]$ and $p \in [0, 1]$. The biological rationale underlying each parameter range is detailed in \textbf{Appendix B}. The goodness-of-fit was evaluated using the coefficient of determination ($R^2$) at the population level. As shown in Fig. \ref{Fig3}A, individual $R^2$ values for the eight subjects were 0.978, 0.968, 0.983, 0.986, 0.978, 0.963, 0.982 and 0.997, respectively. Furthermore, the distribution of $R^2$ across all virtual parameter ensembles (Fig. \ref{Fig3}B) demonstrates consistently high goodness-of-fit. Collectively, these results demonstrate that the proposed model achieves high-fidelity recapitulation of the experimental dynamics across individuals and at the ensemble level.
	
	\begin{figure}[htbp]
		\centering
		\centerline{\includegraphics[width=0.9\textwidth]{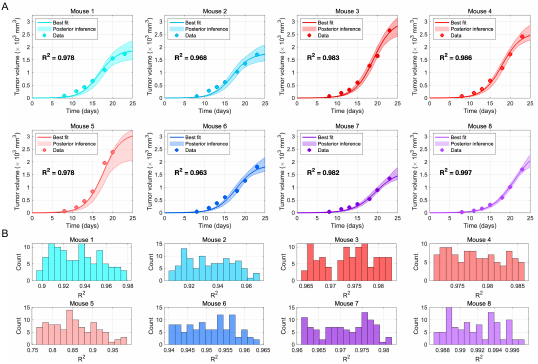}}
		\caption{\textbf{Population-level dynamics of tumor growth and macrophage polarization.} (A) Numerical simulation range of tumor cell counts over time, compared with established experimental time-course data \citep{Reda.NatCommun.2022}. (B) The probability density distributions of the coefficient of determination ($R^2$).}
		\label{Fig3}
	\end{figure}
	
	To further investigate inter-group dynamical profiles, we quantified tumor burden using the area under the curve (AUC), defined as 
	\begin{equation*}
		\mathrm{AUC}_{C} = \int_{t_0}^{t_f} C(t) \mathrm{d}t,
	\end{equation*}
	where $t_0=0$ day and $t_f=25$ day represent the start and end times, respectively. Based on the simulated area under the curve (AUC) distributions, the mice segregated into three distinct phenotypic clusters (Fig. \ref{Fig4}A): Cluster I, consisting of mice 3, 4 and 5; Cluster II, comprising mice 1, 2 and 6; and Cluster III, formed by mice 7 and 8. Data are presented as the mean AUC across 100 \textit{in silico} replicates per group, with error bars indicating the standard deviation.
	
	We further analyzed the macrophage dynamics across the eight \textit{in silico} mouse cohorts. The mean kinetic trajectories for each group ($n$ = 100 virtual mice per group) are shown in Fig. \ref{Fig4}B-D. Comparison among the three phenotypic clusters revealed that mice with low tumor burden (purple) exhibited significantly higher levels of M1-type macrophages than those with high tumor burden (red). The M3-type macrophage population displayed delayed activation kinetics, characterized by a later peak and a more sustained elevation over time. In contrast, M2-type macrophages remained a relatively low level during early phases but increased progressively in the later stage, a pattern attributable to the phenotypic conversion of the persistently elevated M3-type subset into the M2-type state.
	
	Additionally, the mean values and standard deviations of eight key model parameters across the three phenotypic clusters are summarized in Fig. \ref{Fig4}E. Analysis shows that the low tumor burden cluster has a lower tumor growth rate ($\alpha$). This suggests that reduced proliferative capacity corresponds to lower overall tumor load. In contrast, the high burden cluster displays a higher carrying capacity ($K$). This implies that a larger permissible tumor volume may enable further expansion. Meanwhile, tumor burden is inversely related to the baseline level of resting macrophages ($M_0$). It also shows an inverse relationship with the polarization rate ($\kappa$) toward the M1-type phenotype. Higher values of these parameters enhance tumor suppression by promoting the polarization of anti-tumor M1-type macrophages. Conversely, the transition rates $\beta_{13}$ (M1 to M3) and $\beta_{32}$ (M3 to M2) are positively correlated with tumor burden. A higher $\beta_{13}$ promotes the conversion of tumor-suppressive M1-type macrophages into the hybrid M3-type state. A higher $\beta_{32}$ favors the shift from M3-type to pro-tumorigenic M2-type macrophages. Both mechanisms contribute to increased tumor load. The tumor killing rate by M1-type macrophages ($\eta$) is notably higher in the low burden cluster. This reflects more efficient tumor cell elimination, which underpins the reduced tumor burden. Finally, the proportion of the M3-type subset that actively promotes tumor expansion ($p$) is lower in the low burden group. This indicates that the pro‑tumorigenic activity of M3‑type macrophages is limited in this phenotype.
	
	\begin{figure}[htbp]
		\centering
		\centerline{\includegraphics[width=0.9\textwidth]{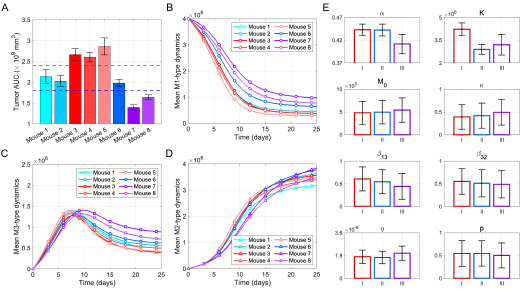}}
		\caption{\textbf{Analysis of inter-group tumor and macrophage dynamics.} (A) Simulated tumor burden, quantified as the area under the curve (AUC), for eight \textit{in silico} mouse cohorts. Data represent mean AUC per group (n = 100 virtual mice) with standard deviation. (B)-(D) Mean temporal trajectories of (B) M1‑, (C) M3‑ and (D) M2‑type macrophage populations across the three identified phenotypic clusters. (E) Mean values and standard deviations of eight key model parameters across the three phenotypic clusters.}
		\label{Fig4}
	\end{figure}
	
	\subsection{Influence of macrophage heterogeneity on tumor-macrophage interactions} \label{sec:sec3_5}
	
	To investigate the impact of macrophage phenotypic heterogeneity on tumor evolution dynamics, we defined two quantitative metrics for analyzing eight \textit{in silico} mouse cohorts. (1) M1-type macrophage burden, quantified as the area under the curve (AUC) of M1-type macrophage levels over the observation window:
	$$\mathrm{AUC}_{M_1} = \int_{t_0}^{t_f} M_1(t)\mathrm{d}t ,$$
	which reflects cumulative anti‑tumor activity. (2) Peak time of M3-type macrophages, defined as the time point at which the M3 population first reaches its maximum within the observed interval:
	$$t^{peak}_{M_3} = \min \left\{t: M_3(t)=\sup _{t_0 \leq \tau \leq t_f} M_3(\tau)\right\},$$
	which captures the delay in hybrid‑phenotype activation.
	
	At the population level, comparative analysis among the three phenotypic clusters (Fig. \ref{Fig5}A and D) showed that the low tumor burden group (purple) exhibited significantly higher late-phase M1-type macrophage burden and a markedly delayed M3-type macrophage peak time compared to the high burden group (red). To quantify these relationships, we computed the correlation between the group mean of each metric and the overall tumor burden (expressed as tumor AUC) across the eight subgroups. Fig. \ref{Fig5}B and E illustrate these associations, with Pearson correlation coefficients of -0.91 and -0.95, respectively, indicating strong negative correlations between tumor AUC and both M1 burden and M3 peak time.
	
	\begin{figure}[htbp]
		\centering
		\centerline{\includegraphics[width=0.9\textwidth]{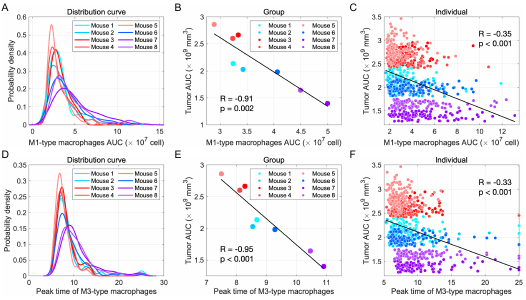}}
		\caption{\textbf{Influence of macrophage phenotypic heterogeneity on tumor-macrophage interactions}. (A) Probability density distribution of the cumulative M1-type macrophage burden, quantified as the area under the curve (AUC). (B)-(C) Correlation between overall tumor burden (AUC of tumor cells) and M1-type macrophage burden at the (B) group level (mean across phenotypic clusters) and (C) individual level (800 virtual mice). Pearson’s correlation coefficient is indicated. (D) Probability density distribution of the time to peak of the M3-type macrophage population. (E)-(F) Correlation between tumor burden and M3 peak time at the (E) group and (F) individual level.}
		\label{Fig5}
	\end{figure}
	
	To validate these trends at the individual level, scatter plots incorporating all 800 virtual mice (Fig. \ref{Fig5}C and F) were generated. Consistent with the population-level findings, individual-level analysis revealed a clear negative correlation between tumor AUC and each of the two macrophage kinetic metrics.
	
	\subsection{Survival analysis}
	
	To achieve a more precise assessment of tumor survival prognosis, we developed a risk stratification model centered around a threshold parameter $Q$. The optimal threshold was determined using the Particle Swarm Optimization (PSO) algorithm. The time of death for virtual sample $i$ is defined as:
	\begin{equation*}
		t_{\mathrm{death} }^{(i)} (V) = \min \{  t \in [t_0, t_f]: V^{(i)} >Q \},
	\end{equation*}
	where $t_0 = 0$ and $t_f = 30$ represent the start and end of the evaluation period, respectively. To account for stochastic mortality effects induced by tumor saturation, the stochastic death risk $V^{(i)}(t)$ at time $t$ for sample $i$ is defined as:
	\begin{equation}
		V^{(i)}(t) = \frac{C^{(i)}(t)}{K^{(i)}} + a \cdot \xi,
	\end{equation}
	where $C^{(i)}(t)$ is the tumor cell count at time $t$, $K^{(i)}$
	is the environmental carrying capacity, $\xi \sim \mathcal{N} (0,1)$ is standard Gaussian white noise introducing stochastic perturbation, and $a = 0.025$  is a noise scaling factor modulating the intensity of randomness. This formulation combines the deterministic saturation ratio $\frac{C^{(i)}(t)}{K^{(i)}}$ with a stochastic perturbation term, thereby integrating fluctuations in mortality risk due to factors such as microenvironmental heterogeneity.
	
	Based on the death time data \citep{Reda.NatCommun.2022}, we constructed a survival function to characterize population survival patterns. The survival rate at time $t$ is defined as the proportion of samples still alive:
	\begin{equation}
		S(t; Q) = \frac{1}{n} \sum_{i=1}^n \mathbb{I}\left(t_{\text{death}}^{(i)}(Q) > t\right) \times 100
	\end{equation}
	where $n$ is the size of the virtual cohort and $\mathbb{I}(\cdot)$ denotes the indicator function. The resulting simulated survival curve was compared with experimental data $D(t)$. The optimization objective was to identify the threshold $Q$ that minimizes the discrepancy between the simulated and experimental curves. We employed the Root Mean Square Error (RMSE) as the quantitative metric:
	\begin{equation}
		\text{RMSE}(Q) = \sqrt{\frac{1}{m} \sum_{j=1}^{m} \left( S(t_j; Q) - D(t_j) \right)^2},
	\end{equation}
	where $m$ is the number of time points in the experimental data. Thus, the problem can be reduced to a univariate optimization defined as:
	\begin{equation}
		Q^* = \arg \min_{Q} \text{RMSE}(Q).
	\end{equation}
	
	To globally optimize the threshold parameter $Q$, PSO was implemented with 20 particles initialized randomly within the plausible range [$Q_{\min}$, $Q_{\max}$] = [0.9, 1]. We assigned random initial velocities. In each iteration, the state of every particle is updated in two steps:
	(1) Velocity update:
	\begin{equation}
		v_i^{k+1} = \omega v_i^k + c_1 r_1 \left(P^{\text{best}}_i - Q_i^k\right) + c_2 r_2 \left(G^{\text{best}} - Q^k\right),
	\end{equation}
	where $Q^k_i$ and $v^k_i$
	denote the position and velocity of particle $i$ at iteration $k$, respectively. The inertia weight 
	$\omega = 0.7$  balances global exploration and local exploitation, the learning factors $c_1 = c_2 = 2$ scale the influence of individual and social experience, $r_1$ and $r_2$ are random numbers uniformly drawn from [0, 1], $P^{\text{best}}_i$ records the particle's personal best position, and 
	$G^{\text{best}}$ is the global best position found by the entire swarm. (2) Position update:
	\begin{equation}
		Q_i^{k+1} = Q_i^{k} + v_i^{k+1}.
	\end{equation}
	After updating the positions, the fitness (RMSE value) corresponding to each particle's $Q$ is evaluated, and the personal and global best records are updated accordingly. Given the relatively narrow prior search range, the algorithm was configured with a maximum of 50 iterations. The final output $G^{\text{best}}$ corresponds to the optimized threshold $Q^*$ that minimizes the RMSE.
	
	To quantitatively characterize tumor survival features, we analyzed the cohort of 800 \textit{in silico} mice described earlier. The survival dynamics of the simulated population closely align with experimental observations (Fig. \ref{Fig6}A). The convergence profile of the Particle Swarm Optimization (PSO) algorithm across iterations confirms that the selected number of iterations ensured stable parameter optimization (Fig. \ref{Fig6}B). Temporal trends in individual death risk stratification are presented in Fig. \ref{Fig6}C. The optimal threshold Q, determined by minimizing the root mean square error (RMSE), supports the biological and statistical plausibility of the model's risk-stratification threshold (Fig. \ref{Fig6}D).
	
	\begin{figure}[htbp]
		\centering
		\centerline{\includegraphics[width=0.9\textwidth]{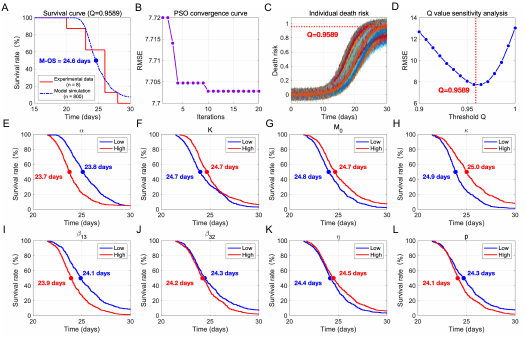}}
		\caption{\textbf{Survival analysis informed by Particle Swarm Optimization.} (A) Survival curve generated from \textit{in silico} simulations, compared against experimental survival data. (B) Convergence profile of the Particle Swarm Optimization (PSO) algorithm, measured by root mean square error (RMSE) over successive iterations. (C) Temporal evolution of the simulated individual death risk, incorporating deterministic saturation and stochastic noise. (D) Dependence of RMSE on the mortality-risk threshold parameter $Q$, with the minimum identifying the optimized threshold $Q^*$. (E)-(L) Comparison of mean survival rates between low‑ and high‑tumor‑burden groups, stratified by eight key model parameters: $\alpha$, $K$, $M_0$, $\kappa$, $\beta_{13}$, $\beta_{32}$, $\eta$ and $p$.}
		\label{Fig6}
	\end{figure}
	
	To evaluate survival outcomes, the cohort (n=800) was grouped by mean parameter values and analyzed with survival curves. The results are described below.
	\begin{itemize}
		\setlength{\itemsep}{2pt}
		\setlength{\parskip}{2pt} 
		\item Tumor growth rate ($\alpha$): An increased $\alpha$ drives rapid disease progression and predicts poorer clinical outcomes (Fig. \ref{Fig6}E).
		\item Tumor carrying capacity ($K$): A higher $K$ value indicates a better survival outcome (Fig. \ref{Fig6}F).
		\item Baseline level of resting macrophages ($M_0$) and polarization rate of  M1-type macrophages ($\kappa$): Higher levels of $M_0$ and $\kappa$ can further promote the polarization of M1-type macrophages with anti-tumor functions, thereby enhancing immune-mediated tumor surveillance and suppression (Fig. \ref{Fig6}G and H).
		\item Macrophage phenotypic transition rates ($\beta_{13}$ and $\beta_{32}$):  $\beta_{13}$ and $\beta_{32}$ accelerate the conversion of anti‑tumor M1-type  macrophages to M3 and then M2 phenotypes. Thereby, they reshape the tumor microenvironment toward immunosuppression, driving tumor progression and lowering survival (Fig. \ref{Fig6}I and J).
		\item Tumor killing rate by M1-type macrophages ($\eta$): A higher value of $\eta$ indicates that more tumor cells can be eliminated by M1‑type macrophages per unit time. This directly counteracts the proliferative advantage of tumor cells, thereby reducing the net tumor growth rate. Clinically, this is reflected as effective control of tumor burden and a favorable survival prognosis (Fig. \ref{Fig6}K).
		\item $M_3$ proportion that promotes tumor expansion ($p$): An increased proportion of pro‑tumor phenotype in M3‑type macrophages predicts poorer prognosis (Fig. \ref{Fig6}L).
	\end{itemize}
	
	Building upon the eight \textit{in silico} mouse cohorts generated in Section \ref{sec:sec3_4}, this study investigated the influence of macrophage phenotypic characteristics on host survival by analyzing two key dynamical metrics: (1) the M1 macrophage burden, quantified as the area under the curve (AUC) of M1-type macrophages over the observation window (0-25 days), and (2) the time to peak of the M3-type macrophage population. The temporal profiles of the mean survival rates for the eight cohorts are shown in Fig. \ref{Fig7}. The results indicate that the low tumor burden group (purple) exhibited significantly higher mean survival than the high burden group (red), consistent with reduced systemic damage under lower tumor loads.
	
	At the population level, we further analyzed the correlation between these metrics and the overall tumor burden (expressed as tumor AUC) across the three phenotypic clusters. As shown in Fig. \ref{Fig7}B and D, both M1-type macrophage burden and M3-type peak time showed strong negative correlations with tumor AUC (Pearson's $r=-0.84$ for each). This finding suggests that a higher M1-type macrophage burden may suppress tumor progression by enhancing anti-tumor immune responses, while a delayed M3 peak could attenuate the subsequent conversion into pro-tumorigenic M2-type macrophages, thereby impeding tumor growth. To validate this relationship at the individual level, scatter plots including all 800 virtual mice were generated (Fig. \ref{Fig7}C and E). Individual-level analysis confirmed a consistent negative trend between tumor AUC and each of the two macrophage kinetic metrics, reinforcing the population-based conclusions.
	
	\begin{figure}[htbp]
		\centering
		\centerline{\includegraphics[width=0.9\textwidth]{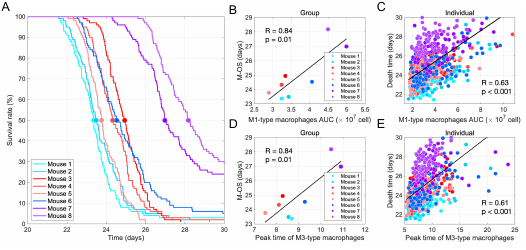}}
		\caption{\textbf{Influence of macrophage phenotypic heterogeneity on survival outcome.} (A) Time-evolution of the mean survival rate for each of the eight \textit{in silico} mouse cohorts. (B)-(C) Correlation between median overall survival and cumulative M1-type macrophage burden (area under the curve) at the (B) group level (mean across phenotypic clusters) and (C) individual level (all 800 virtual mice). (D)-(E) Correlation between median overall survival and the time to peak of the M3-type macrophage population at the (D) group and (E) individual level. Pearson correlation coefficients are provided for each association.}
		\label{Fig7}
	\end{figure}

	\section{Discussion} \label{sec:sec4}
	
	The ongoing interplay between macrophages and tumor cells generates a dynamic network that continuously shapes the evolutionary trajectory of cancer. To elucidate the underlying mechanisms of this interplay, we established a mechanistic mathematical model of tumor-macrophages interactions. Theoretical analysis indicates that the tumor-free equilibrium is asymptotically stable when the condition tumor growth rate $\alpha < \alpha^*$ is satisfied. Biologically, this inequality defines a critical boundary determining whether the tumor is controlled or progresses in the absence of clinical intervention. When the tumor growth rate falls below this threshold, the system effectively eliminates the tumor, leading to recovery; otherwise, immune escape occurs, resulting in a tumorous state. Model parameters were inferred using a Markov Chain Monte Carlo (MCMC) framework, achieving a high-fidelity fit to longitudinal tumor volume data ($R^2 > 0.96$). The model accurately recapitulates the established temporal dynamic characteristics of three TAM subsets (anti-tumor M1, pro-tumor M2 and hybrid M3 phenotypes), consistent with prior experimental and theoretical studies. To examine the influence of $M_1$‑type macrophage polarization and the baseline level of resting macrophages on tumor‑macrophage dynamics, we conducted a bifurcation analysis of the system described by system \eqref{Model} using numerical continuation methods. Notably, the analysis reveals bistability between tumor-free and tumor-dominant regimes. Clinically, this bistability provides a plausible mechanistic explanation for the heterogeneous progression patterns observed in cancer. Furthermore, we applied the Sobol' sensitivity method to quantify the influence of model parameters on macrophage polarization dynamics. Global sensitivity analysis reveals time-evolving regulatory logic: tumor control transitions from proliferation rate to environmental carrying capacity, while macrophage dynamics are governed by sequential dominance of polarization, transformation and cross-regulatory signals.
	
	Subsequent population-level analysis, informed by ABC, identified three distinct tumor burden clusters among eight mouse groups. These clusters exhibited significant differences in key parameters, including tumor growth rate, macrophage polarization efficiency and phenotypic transition rates. \textit{In silico} ensembles recapitulated individual and group-level tumor-macrophage dynamics with high fidelity. Low tumor burden clusters exhibited elevated M1 activity, delayed M3 activation, and specific parameter profiles linked to enhanced anti-tumor function. Furthermore, we quantified the impact of macrophage heterogeneity on survival prognosis. The M1-type macrophage burden and the time to peak of the M3 population both showed strong negative correlations with overall tumor load. We further identify two key prognostic metrics: cumulative M1-type macrophage activity and delayed M3-type macrophage peak, both strongly negatively correlated with tumor burden and positively associated with survival. Analysis of eight key parameters reveals that higher tumor burden is associated with poorer survival when driven by increased growth rate ($\alpha$), but paradoxically better survival in scenarios of high carrying capacity ($K$) or elevated $M_0/\kappa$, suggesting context-dependent immune modulation. Low burden favors survival when linked to higher $\beta_{13}/\beta_{32}$ or pro-tumorigenic M3 fraction ($p$), highlighting the complex, phenotype-specific role of macrophage dynamics in shaping survival outcomes. These results demonstrate that our model captures the critical, stage-specific crosstalk between tumor progression and macrophage phenotypic plasticity, providing a quantitative framework for understanding tumor-immune dynamics.
	
	The present model offers a valuable theoretical framework for analyzing tumor-macrophage interactions and holds potential relevance for therapeutic applications. However, it inherently possesses several limitations: (1) The model fails to account for spatial heterogeneity in the distribution of distinct cell populations. Spatial architecture plays a critical role in mediating cell-cell contact and migratory behaviors, and there has been growing interest in spatially explicit models of tumor-macrophage crosstalk \citep{Suveges.BullMathBiol2020,Suveges.FrontApplMathStat2022,Lai.PNAS.2018,Friedman.DCDSB.2019}. (2) Intratumoral heterogeneity is not considered, despite its potential to drive drug resistance and treatment failure. Recent years have witnessed increasing efforts to mathematically model such heterogeneity \citep{LiuApplMathModel2022,WangPLoSComputBiol2025,Cunningham2018}. (3) The model overlooks other key immune constituents of the tumor microenvironment. For instance, CD8+ T cells and regulatory T cells exert immunostimulatory and immunosuppressive influences, respectively, and play crucial roles in tumor progression. Mathematical models exploring interactions between tumors and various immune cells have thus attracted considerable attention \citep{Friedman.BullMathBiol.2020,Liao.MathBiosci.2023,Liao.MathBiosci.2024}. Future work should prioritize the development of spatially explicit models-such as those based on partial differential equations or agent-based models incorporate tissue-level spatial organization. Moreover, extending the model to incorporate phenotypic plasticity underlying resistance mechanisms, as well as a broader repertoire of immune effectors such as CD8+ T cells and regulatory T cells, will be essential for capturing the complexity of the tumor-immune system and generating more clinically actionable insights.
	
	\section*{Declaration of Competing Interest}
	The authors declare there is no conflict of interest.
	
	\section*{Acknowledgements}
	This work was supported by the National Key R\&D Program of China (Grant No. 2021YFA0719200) and the National Natural Science Foundation of China (12301617).
	
	\section*{Appendix A. Proofs of Theorems \ref{thm2.1}-\ref{thm2.1_e} }
	\label{app_A}
	
	\textbf{Proofs of Theorem \ref{thm2.1}}
	\begin{proof} Let $(C(t), M_1(t), M_2(t), M_3(t) )$ be a solution of system \eqref{Model} with initial condition $(C(0), M_1(0), M_2(0), M_3(0) )$. Assume $C(0)>0$, $M_1(0)>0$, $M_2(0)>0$ and $M_3(0)>0$.
		
		From the first equation of \eqref{Model}, we have
		\begin{equation*}
			C(t) = C(0) {\mathrm e}^{ \int_{ 0}^{t} g(s)  {\rm d} s},
		\end{equation*}
		where
		\begin{equation*}
			g(s) = \alpha \left(1+\delta (M_2(s) + p M_3(s) ) \right) \left( 1 - \frac{C(s) }{K}\right) - \eta \left( M_1(s) +(1-p) M_3(s)  \right).
		\end{equation*}
		Since the exponential function is always positive, it follows that $C(t) > 0$ for all $t \geq 0$. 
		
		We now prove that $M_1(t) >0$, $M_2(t) >0$ and $M_3(t) >0$ for all $t \geq 0$. To establish this statement, we argue by contradiction. Suppose that there exists a first time $t_0 >0$ such that 
		\begin{equation*}
			\min \{ M_1(t_0), M_2(t_0), M_3(t_0) \} =0.
		\end{equation*}
		We analyze all possible scenarios:
		
		Case 1. $M_1(t_0) =0$, with $M_1(t) >0$ for $t \in [0, t_0)$, $M_2(t_0) \geq 0$ and $M_3(t_0) \geq 0$. Then $\frac{\mathrm{d} M_1(t_0)}{\mathrm{d} t} \leq 0$. However, from the second equation of system \eqref{Model}, we have 
		\begin{equation*}
			\frac{\mathrm{d} M_1(t_0)}{\mathrm{d} t} = \kappa M_0 - \beta_{13}\frac{C(t_0)}{K_1+C(t_0)}M_1(t_0) + \beta_{31} M_3(t_0) - d_1 M_1(t_0) = \kappa M_0 + \beta_{31} M_3(t_0) >0,
		\end{equation*}
		which contradicts $\frac{\mathrm{d} M_1(t_0)}{\mathrm{d} t} \leq 0$.
		
		Case 2. $M_2(t_0) =0$, with $M_2(t) >0$ for $t \in [0, t_0)$, $M_1(t_0) \geq 0$ and $M_3(t_0) \geq 0$. A similar argument as in Case 1 leads to a contradiction.
		
		Case 3. $M_3(t_0) =0$, with $M_3(t) >0$ for $t \in [0, t_0)$, and $M_1(t_0) = 0, M_2(t_0) > 0$ or $M_1(t_0) > 0, M_2(t_0) = 0$ or $M_1(t_0) = M_2(t_0) = 0$. Without loss of generality, assuming that $M_1(t_0) =0, M_2(t_0) > 0$. Then $\frac{\mathrm{d} M_1(t_0)}{\mathrm{d} t} \leq 0$. On the basis of the second equation of system \eqref{Model}, we obtain
		\begin{equation*}
			\frac{\mathrm{d} M_1(t_0)}{\mathrm{d} t} = \kappa M_0 - \beta_{13}\frac{C(t_0)}{K_1+C(t_0)}M_1(t_0) + \beta_{31} M_3(t_0) - d_1 M_1(t_0) = \kappa M_0 >0,
		\end{equation*}
		which contradicts $\frac{\mathrm{d} M_1(t_0)}{\mathrm{d} t} \leq 0$.
		
		Case 4. $M_3(t_0) =0$, with $M_3(t) >0$ for $t \in [0, t_0)$, $M_1(t_0) > 0$ and $M_2(t_0) > 0$. Then $\frac{\mathrm{d} M_3(t_0)}{\mathrm{d} t} \leq 0$. From the fourth equation of system \eqref{Model}, we obtain 
		\begin{eqnarray*}
			\frac{\mathrm{d} M_3(t_0)}{\mathrm{d} t} & = & \beta_{13}\frac{C(t_0)}{K_1+C(t_0)}M_1(t_0)  - \beta_{31}M_3(t_0) - \beta_{32}\frac{C(t_0)}{K_2+C(t_0)}M_3(t_0) + \beta_{23}M_2(t_0) - d_3 M_3(t_0)  \\
			& =&  \beta_{13}\frac{C(t_0) }{K_1+C(t_0)}M_1(t_0)   +  \beta_{23}M_2(t_0)  \\
			& >& 0,
		\end{eqnarray*}
		which is again a contradiction.
		
		In all cases, the assumption that a component first becomes non-positive leads to a contradiction. Therefore, $M_1(t), M_2(t), M_3(t) > 0$ for all $t \geqslant 0$, and the proof is complete.
	\end{proof}	
	
	\textbf{Proofs of Theorem \ref{thm2.1_b}}
	\begin{proof}
		By Theorem \ref{thm2.1}, the solutions of system \eqref{Model} satisfy $C(t) > 0$, $M_1(t) > 0$, $M_2(t) > 0$ and $M_3(t) > 0$ for all $t \geqslant 0$. 
		
		We begin by proving that $C(t)$ remians bounded for $t \geqslant 0$. According to the first equation of system \eqref{Model}, we obtain
		\begin{equation*}
			\begin{aligned}
				\dfrac{{\rm d} C(t)}{{\rm d} t} \big |_{ \{C  = K, M_1, M_2, M_3 > 0 \} } = - \eta \left ( M_1 + (1-p)M_3 \right ) C   \leqslant 0,
			\end{aligned}
		\end{equation*}
		which implies there exists a $t_1 >0$ such that $C(t) \leqslant K $ holds for all $t \geqslant t_1$. Consequently, $C(t)$ is bounded for all $t \geqslant 0$.
		
		To bound the macrophage populations, let $d_M := \min \left \{ d_1 , d_2, d_3  \right \}$. From the second, third and fourth equations of system \eqref{Model}, we derive
		\begin{eqnarray*}
			\dfrac{{\rm d} ( M_1 + M_2 + M_3 ) }{{\rm d} t} & \leq &  \kappa M_0 + \gamma \left ( 1+\lambda \frac{C}{K_0+C} \right ) M_0 - d_{M_1} M_1 -  d_{M_2} M_2  -  d_{M_3} M_3 \\
			& \leq &   \kappa M_0 + \gamma \left ( 1+\lambda \right ) M_0 - d_M M_1 - d_M M_2 - d_M M_3 \\
			& =& ( \kappa + \gamma + \gamma \lambda ) M_0  - d_M ( M_1 + M_2 + M_3 ).
		\end{eqnarray*}
		This inequality implies that for all $t \geqslant 0$,
		\begin{eqnarray*}
			M_1(t) + M_2(t) + M_3(t) \leq \hat{M} := \max \left \{ M_1(0) + M_2(0) + M_3(0), \dfrac{ (\kappa + \gamma(1+\lambda)) M_0 }{ d_M}  \right \}
		\end{eqnarray*}
		and hence each component satisfies $ M_i(t) \leq \hat{M}$ for $i=1,2,3$.
		This completes the proof.
	\end{proof}
	
	\textbf{Proofs of Theorem \ref{thm2.1_c}}
	
	\begin{proof}
		The tumor-free equilibria of system \eqref{Model} are determined by solving the following algebraic system:
		\begin{equation}
			\label{tre}
			\left\{
			\begin{aligned}
				\kappa M_0 + \beta_{31}M_{30}^* - d_1 M_{10}^* &=& 0, \\
				\gamma M_0 - \beta_{23}M_{20}^* - d_2 M_{20}^* &=& 0, \\
				- \beta_{31}M_{30}^*  + \beta_{23}M_{20}^* - d_3 M_{30}^* &=& 0.
			\end{aligned}
			\right.
		\end{equation}
		
		From the second equation of \eqref{tre}, we obtain
		\begin{equation}
			\label{tre_1}
			M_{20}^* = \dfrac{ \gamma M_0}{\beta_{23}+d_2}.
		\end{equation}
		Substituting \eqref{tre_1} into the third equation of \eqref{tre}, we have
		\begin{equation}
			\label{tre_2}
			M_{30}^* = \dfrac{ \beta_{23} }{\beta_{31}+d_3} M_{20}^* =  \dfrac{\gamma M_0  \beta_{23} }{(\beta_{23}+d_2) (\beta_{31}+d_3)} .
		\end{equation}
		Finally, substituting the equation \eqref{tre_2} into the first equation of \eqref{tre} gives
		\begin{equation}
			\label{tre_3}
			M_{10}^* = \dfrac{ \kappa M_0 + \beta_{31}M_{30}^* }{ d_1} = \dfrac{ \kappa M_0 (\beta_{23}+d_2) (\beta_{31}+d_3)  + \beta_{31} \gamma M_0  \beta_{23} }{ d_1 (\beta_{23}+d_2) (\beta_{31}+d_3) } .
		\end{equation}
		This completes the proof.
	\end{proof}
	
	\textbf{Proofs of Theorem \ref{thm2.1_d}}
	\begin{proof}
		The Jacobian matrix at the tumor-free equilibrium is given by
		\begin{equation}
			\label{eq:j2}
			J= \left( 
			\begin{array}{cccc}
				\alpha \left(1+ \delta ( M_{20}^* + p  M_{30}^*) \right) - \eta \left(  M_{10}^* +(1-p) M_{30}^* \right)  & 0   & 0  & 0      \\
				-\dfrac{\beta_{13}M_{10}^* }{ K_1 }  & -d_1 & 0  & \beta_{31}       \\
				\dfrac{ \gamma \lambda M_0 }{ K_0 } + \dfrac{\beta_{32}M_{30}^* }{ K_2 } & 0 & -\beta_{23} -d_2 & 0       \\
				\dfrac{\beta_{13}M_{10}^* }{ K_1 } - \dfrac{\beta_{32}M_{30}^* }{ K_2 } & 0 & \beta_{23} & -\beta_{31} -d_3          
			\end{array}
			\right),
		\end{equation}
		Therefore, all eigenvalues of the matrix $J$ are $ \alpha \left(1+ \delta ( M_{20}^* + p  M_{30}^*) \right) - \eta \left(  M_{10}^* +(1-p) M_{30}^* \right) $, $-d_1 $, $-\beta_{23} -d_2 $ and $-\beta_{31} -d_3 $. Under condition \eqref{eq:ps_4aaa}, all the eigenvalues of the Jacobian matrix $J$ are all negative numbers. Thereby, the tumor-free equilibrium $\left( 0, M_{10}^*, M_{20}^*, M_{30}^* \right)$ is locally asymptotically stable. This completes the proof.
	\end{proof}

	\textbf{Proofs of Theorem \ref{thm2.1_e}}
	\begin{proof}
		The positive steady states of system \eqref{Model} satisfy following algebraic equation:
		\begin{equation}
			\label{pss}
			\left\{
			\begin{aligned}
				\alpha \left \{ 1+\delta\left ( M_2^*+p M_3^* \right )  \right \} \left ( 1 - \frac{C^*}{K} \right )  - \eta \left ( M_1^* + (1-p)M_3^* \right ) & = 0,\\
				\kappa M_0 - \beta_{13}\frac{C^*}{K_1+C^*}M_1^* + \beta_{31}M_3^* - d_1 M_1^* & = 0,\\
				\gamma \left( 1+\lambda \frac{C^*}{K_0+C^*} \right ) M_0 + \beta_{32}\frac{C^*}{K_2+C^*}M_3^* - \beta_{23}M_2^* - d_2 M_2^* & =0, \\
				\beta_{13}\frac{C^*}{K_1+C^*}M_1^* - \beta_{31}M_3^* - \beta_{32}\frac{C^*}{K_2+C^*}M_3^* + \beta_{23}M_2^*- d_3 M_3^* & = 0.
			\end{aligned}
			\right.
		\end{equation}
		Using Cramer's rule on the last three equations of \eqref{pss}, we obtain
		\begin{equation}
			\label{pss_1}
			M_1^* = \dfrac{f_1(C^*) }{f_4(C^*)}, ~~~~ M_2^* = \dfrac{f_2(C^*) }{f_4(C^*)}, ~~~~ M_3^* = \dfrac{f_3(C^*) }{f_4(C^*)},
		\end{equation}
		where
		\begin{equation}
			\label{pss_2}
			\begin{aligned}
				f_1(C^*)  = -\kappa M_0 \left( d_2 \left( \beta_{31} + \beta_{32} \frac{C^*}{K_2+C^*} +d_3 \right) + \beta_{23} (\beta_{31} + d_3)  \right) -  \beta_{31}  \beta_{23} \gamma M_0 \left( 1+\lambda \frac{C^*}{K_0+C^*} \right),
			\end{aligned}
		\end{equation}
		
		\begin{equation}
			\label{pss_3}
			\begin{aligned}
				f_2(C^*)  = & - \gamma M_0 \left( 1+\lambda \frac{C^*}{K_0+C^*} \right) \left( \beta_{13}\frac{C^*}{K_1+C^*} \left( \beta_{32} \frac{C^*}{K_2+C^*} +d_3 \right) + d_1 \left(  \beta_{31} + \beta_{32} \frac{C^*}{K_2+C^*} +d_3 \right) \right) \\
				& - \beta_{13}\frac{C^*}{K_1+C^*} \kappa M_0 \beta_{32} \frac{C^*}{K_2+C^*} ,
			\end{aligned}
		\end{equation}
		
		\begin{equation}
			\label{pss_4}
			\begin{aligned}
				f_3(C^*)  =  - \left( \beta_{13}\frac{C^*}{K_1+C^*}+d_1 \right) \beta_{23} \gamma M_0 \left( 1+\lambda \frac{C^*}{K_0+C^*} \right) - \beta_{13} \frac{C^*}{K_1+C^*} \kappa M_0 (\beta_{23} +d_2) 
			\end{aligned}
		\end{equation}
		and 
		\begin{equation}
			\label{pss_5}
			\begin{aligned}
				f_4(C^*)  =  - \left( \beta_{13}\frac{C^*}{K_1+C^*}+d_1 \right) \left(  d_2 \left(  \beta_{32} \frac{C^*}{K_2+C^*} +d_3  \right) + \beta_{23} d_3 \right) - d_1 \beta_{31} (\beta_{23} + d_2).
			\end{aligned}
		\end{equation}
		
		Set 
		\begin{equation}
			\label{eq:ps_4aa}
			\begin{aligned}
				f(C^*) & =  \alpha \left \{ 1+\delta\left ( M_2^*+p M_3^* \right )  \right \} \left ( 1 - \frac{C^*}{K} \right )  - \eta \left ( M_1^* + (1-p)M_3^* \right ) \\
				& =  \alpha \left \{ 1+\delta\left ( \dfrac{f_2(C^*) }{f_4(C^*)}+p \dfrac{f_3(C^*) }{f_4(C^*)} \right )  \right \} \left ( 1 - \frac{C^*}{K} \right )  - \eta \left ( \dfrac{f_1(C^*) }{f_4(C^*)} + (1-p)\dfrac{f_1(C^*) }{f_4(C^*)} \right ). 
			\end{aligned}
		\end{equation} 
		From \eqref{pss_2} - \eqref{pss_5}, for $\beta_{23} \geqslant 0$, the functions $f_i (C^*)>0 (i=1, 2, 3, 4) $ on the interval $[0, K]$, ensuring $M_i^*>0 (i=1, 2, 3)$ on the same interval. In particular, 
		\begin{equation}
			f(K) = - \eta \left ( M_1^* + (1-p)M_3^*\right) <0. 
		\end{equation}
		In addition, under condition \eqref{eq:ps_4aa0},
		\begin{equation}
			\label{eq:ps_4aa1}
			\begin{aligned}
				f(0) = \alpha \left(1+ \delta ( M_{20}^* + p  M_{30}^*) \right) - \eta \left(  M_{10}^* +(1-p) M_{30}^* \right) > 0. 
			\end{aligned}
		\end{equation} 
		Thereby, by the Intermediate Value Theorem, there exists at least one root of $f(C^*)=0$ in $(0, K)$, which guarantees the
		existence of a positive steady state for system \eqref{Model}. This concludes the proof.
	\end{proof}

	\section*{Appendix B. Parameter estimation}
	\label{app_B}
	
	In this study, all parameter values and initial conditions used are summarized in Table \ref{tab:table1}. Specific values and their sources are as follows:
	
	\textbf{Initial values.} Model calibration was performed using experimental data from \citep{Reda.NatCommun.2022}. The initial tumor cell count was set to $C(0) = 2\times10^5$ cells, corresponding to the number of tumor cells injected in murine experiments \citep{Reda.NatCommun.2022}. To convert simulated cell counts into tumor volume, we used the relation
	\begin{equation*}
		V=\frac{v_{\text {cell}} \times C}{v_{\text {fraction}}},
	\end{equation*} 
	where $v_{\text{cell}} = 2.572\times10^{-9}$ ${\rm mm^3}$/cell represents the volume of a single tumor cell and $v_{\text{fraction}} = 0.37$ denotes the volume fraction of tumor cells within the total tissue volume \citep{Li.NPJSystBiolAppl.2025}. Consequently, the tumor cell count can be converted to volume by multiplying by a factor of $6.95\times10^{-6}$ mm$^3$/cell. 
	
	Based on prior modeling studies \citep{Shu.ApplMathModel.2020, Li.BullMathBiol.2024, Zhang.TheoryBiosci.2025}, the initial number or steady-state level of macrophages under physiological conditions typically ranges from $4\times10^5$ to $1\times10^8$ cells. We therefore set that the initial population of M1-type macrophages is $M_1(0) = 4 \times 10^6$ cells. In the absence of tumor cell inoculation, the immune system is predominantly composed of M1-type macrophages. Hence, we set the initial populations of both M2- and M3-type macrophages are zero, namely, $M_2(0) = M_3(0) = 0$ cells.
	
	\textbf{Tumor kinetics parameters.} Following \citep{Pillis.JTheorBiol.2006, Khalili.JTheorBiol.2023}, the tumor cell growth rate was set to $\alpha = 0.431$ day$^{-1}$. This value lies within the empirically supported range of $0.38$ to $0.46$ day$^{-1}$ reported in multiple modeling studies \citep{Xue.JTheorBiol.2023, Friedman.BullMathBiol.2018, Das.MathComputSimulat.2022,Rhodes.JTheorBiol.2019}. Tumor carrying capacity can range from $2.17 \times 10^8$ to $5 \times 10^8$ cells \citep{Guo.DCDSB.2023, Mahasa.JTheorBiol.2016, Rhodes.JTheorBiol.2019,Pillis.JTheorBiol.2006}. Thus, the tumor carrying capacity $K$ was taken as $3 \times 10^8$ cells as an intermediate value within the typical interval of $2.17 \times 10^8$ to $5 \times 10^8$ cells. 
	
	In accordance with \citep{Eftimie.JTheorBiol.2017}, where the macrophage carrying capacity is given as $1\times10^9$ cells, the tumor amplification coefficient mediated by M2-type macrophages was set to its reciprocal, $1\times 10^{-9}$ cells$^{-1}$. For the M3-type macrophages, we assumed an equal functional split between pro- and anti-tumor phenotypes, corresponding to $p=0.5$ (where $p$ can theoretically range from 0 to 1).
	
	Based on studies \citep{Shariatpanahi.JTheorBiol.2018,Kreger.CancerImmunolRes.2023}, the tumor killing rate by cytotoxic T cells in the adaptive immune system is $1.1\times 10^{-7}$ cells$^{-1}$day$^{-1}$. Given that macrophages exhibit lower cytotoxic activity than cytotoxic T cells, the macrophage killing rate $\eta$ is set to one-tenth of this value, yielding $\eta = \frac{1}{10} \times 1.1\times 10^{-7} = 1.1\times 10^{-8}$ cells$^{-1}$day$^{-1}$. This value can vary within the range of $1\times10^{-8}$ to $3\times10^{-8}$ cells$^{-1}$day$^{-1}$, reflecting the functional heterogeneity observed in macrophage-mediated tumor killing.
	
	\textbf{Polarization kinetics parameters.} Consistent with prior modeling study \citep{Zhang.TheoryBiosci.2025}, the baseline level of resting macrophages is set to $4 \times 10^5$ cells (i.e., $M_0 = 4 \times 10^5$ cells, with a physiologically plausible range typically between $1\times10^5$ and $1\times10^6$ cells). Reported polarization rate of M1-type macrophages typically range from $0.300$ to $0.500$ day$^{-1}$ \citep{Liao.MathBiosci.2023,Guo.DCDSB.2023}, while the polarization rate of M2-type macrophages vary between $2\times 10^{-2}$ and $1.125$ day$^{-1}$ \citep{Liao.MathBiosci.2023,Eftimie.JTheorBiol.2017}. Accordingly, we selected baseline values of $\kappa = 0.4$ day$^{-1}$ for M1-type macrophages polarization and $\gamma = 0.1$ day$^{-1}$ for M2-type macrophages polarization. To maintain model simplicity while acknowledging that M1 polarization is influenced by diverse microenvironmental factors, we constrained $\kappa$ to a range of $0$ to $1$ day$^{-1}$ in our simulations, avoiding the introduction of more complex nonlinear dynamics.
	
	Based on the polarization characteristics of M2-type macrophages in tumor-immune interactions and their response mechanisms to the tumor microenvironment, we set that the tumor half-saturation constant $K_0$ for M2 polarization to one-tenth of the maximum tumor carrying capacity, i.e., $K_0 = \frac{1}{10}K = 3\times10^7$ cells. Tumor cells are modeled to enhance M2-type macrophages polarization rate by up to 50\%. When $C \gg K_0$, the term $\gamma \left( 1+\lambda \frac{C}{K_0+C} \right)$ approaches $\gamma \left( 1+ \lambda \right)$, leading to the estimate $\lambda = 0.5$ day$^{-1}$.
	
	\textbf{Transformation kinetics parameters.} Based on the observed phenotypic shift in macrophage dominance from the M1 in early stages to the M2 phenotype in later stages of tumor progression \citep{Chen.ClinCancerRes.2011}, we set the half-saturation constants for phenotypic transitions as $K_1 = K_2 = \frac{1}{10}K_0 = 3 \times 10^{6}$ cells. In line with previous mathematical models \citep{Eftimie.JTheorBiol.2021}, macrophage re-polarization rates are set to $\beta_{31} = \beta_{23} = 0$ day$^{-1}$, reflecting the limited plasticity between established phenotypes. In contrast, the transformation rates $\beta_{13}$ and $\beta_{32}$ are assigned a value of $0.5$ day$^{-1}$, capturing more active polarization toward intermediate and M2 phenotypes in response to tumor-derived signals. To further reveal the heterogeneity in macrophage phenotypic transformation, the values of $\beta_{13}$ and $\beta_{32}$ are allowed to vary across the range of 0 to 1 day$^{-1}$ when generating our \textit{in silico} mouse cohort.
	
	\textbf{Death kinetics parameters.} According to clinical experimental data from Wacker et al. \citep{Wacker.VirchowsArchBCellPatholInclMolPathol.1986}, tissue-resident macrophages exhibit a half-life to 12.4 days. Therefore, we set the half-life of macrophages to 12.4 days ($t^{M_i}_{1/2} = 12.4$ days, $i=1,2,3$). Assuming exponential decay of the form $M_i(t)=M_i(0) {\rm e}^{d_i t}$, the half-life condition $\frac{1}{2} M_i(0)=M_i(0) \exp \left(-d_i t_{1/2}^{M_i}\right)$ leads to the death rate $d_i=\frac{\ln 2}{12.4\ \text{day}}=0.056 ~ \text{day}^{-1}$. This value is adopted for all macrophage subtypes in the model.
	
	\begin{table}[h]
		\centering
		\tiny
		\caption{Description of related parameter. }
		\begin{tabular}{ccccc}
			\toprule
			Parameter & Description & Values used & Unit  & Reference and estimation \\
			\midrule
			$C(0)$ & Initial value of tumor cells & $2\times10^5$ & $\rm cells$ & \citep{Reda.NatCommun.2022} \\
			$M_1(0)$ & Initial value of M1-type macrophages & $4\times10^6$ & $\rm cells$ & \citep{Shu.ApplMathModel.2020, Li.BullMathBiol.2024, Zhang.TheoryBiosci.2025} \\
			$M_2(0)$ & Initial value of M2-type macrophages & $0$ & $\rm cells$ & Estimation  \\
			$M_3(0)$ & Initial value of M3-type macrophages & $0$ & $\rm cells$ & Estimation  \\
			$\alpha$ & Tumor growth rate & $0.431$ & $\rm day^{-1}$ & \citep{Pillis.JTheorBiol.2006, Khalili.JTheorBiol.2023} \\
			$K$ & Tumor carrying capacity & $3\times10^8$ & $\rm cells$ &  \citep{Guo.DCDSB.2023, Mahasa.JTheorBiol.2016, Rhodes.JTheorBiol.2019}  \\
			$\delta$ & Tumor amplification coefficient by $M_2$ & $1\times 10^{-9}$ & $\rm cells^{-1}$  & Estimation \& \citep{Eftimie.JTheorBiol.2017} \\
			$p$ & Proportion of $M_3$ that promotes tumor expansion & $0.5$ & $\rm none$  & Estimation \\
			$\eta$ & Killing rate of tumor by $M_1$ & $1.1 \times 10^{-8}$ & $\rm cells^{-1} day^{-1}$ & Estimation \& \citep{Shariatpanahi.JTheorBiol.2018,Kreger.CancerImmunolRes.2023} \\
			$\kappa$ & Polarization rate of $M_1$ & $0.4$ & $\rm day^{-1}$  & \citep{Liao.MathBiosci.2023,Guo.DCDSB.2023} \\
			$\gamma$ & Polarization rate of $M_2$ & $0.1$ & $\rm day^{-1}$ & Estimation \&  \citep{Liao.MathBiosci.2023,Eftimie.JTheorBiol.2017} \\
			$\lambda$ & Polarization amplification coefficient by tumor cells & $0.5$ & $\rm none$ & Estimation \\
			$M_0$ & Baseline level of resting macrophages & $4 \times 10^5$ & $\rm cells$ & Estimation \& \citep{Zhang.TheoryBiosci.2025} \\
			$\beta_{13}$ & Transformation rate from $M_1$ to $M_3$ & $0.5$ & $\rm day^{-1}$ & Estimation  \\
			$\beta_{31}$ & Transformation rate from $M_3$ to $M_1$ & $0$ & $\rm day^{-1}$ & \citep{Eftimie.JTheorBiol.2021}  \\
			$\beta_{23}$ & Transformation rate from $M_2$ to $M_3$ & $0$  & $\rm day^{-1}$ & \citep{Eftimie.JTheorBiol.2021} \\
			$\beta_{32}$ & Transformation rate from $M_3$ to $M_2$ & $0.5$ & $\rm day^{-1}$ & Estimation \\
			$K_0$ & Tumor half-saturation constant in the $M_2$ polarization & $3 \times 10^{7}$ & $\rm cells$ & Estimation  \\
			$K_1$ & Tumor half-saturation constant in the $M_1$ transformation & $3 \times 10^{6}$ & $\rm cells$ & Estimation \\
			$K_2$ & Tumor half-saturation constant in the $M_3$ transformation & $3 \times 10^{6}$ & $\rm cells$ & Estimation \\
			$d_1$ & Death rate of $M_1$ & $0.056$ & $\rm day^{-1}$ & \citep{Wacker.VirchowsArchBCellPatholInclMolPathol.1986} \\
			$d_2$ & Death rate of $M_2$ & $0.056$ & $\rm day^{-1}$ & \citep{Wacker.VirchowsArchBCellPatholInclMolPathol.1986} \\
			$d_3$ & Death rate of $M_3$ & $0.056$ & $\rm day^{-1}$ & \citep{Wacker.VirchowsArchBCellPatholInclMolPathol.1986} \\
			\bottomrule
		\end{tabular}
		\label{tab:table1}
	\end{table}

	\printcredits
	
	\bibliographystyle{cas-model2-names}
	\bibliography{references.bib}
	
\end{document}